\documentclass[11pt,a4paper]{article}
\usepackage[margin=1in]{geometry}

\usepackage{amsmath,amsfonts,amssymb,amsthm}
\usepackage{graphicx}
\usepackage{enumerate}
\usepackage{bbm}
\usepackage{tikz}
\usepackage{verbatim}
\usetikzlibrary{arrows,shapes,positioning,arrows}
\usetikzlibrary{decorations.markings}
\tikzstyle arrowstyle=[scale=1]
\usetikzlibrary{graphs}
\usetikzlibrary{graphs.standard}
\usetikzlibrary{matrix,positioning,fit,backgrounds}
\tikzstyle{directed}=[postaction={decorate,decoration={markings,mark=at position .65 with {\arrow[arrowstyle]{stealth}}}}]
\usetikzlibrary{topaths}
\usepackage[ruled,vlined,linesnumbered]{algorithm2e}
\usepackage{hyperref}
\hypersetup{colorlinks=true,citecolor=blue, linkcolor=blue, urlcolor=blue}

\newtheorem{thm}{Theorem}

\newtheorem{defn}[thm]{Definition}
\newtheorem{lemma}[thm]{Lemma}
\newtheorem{coro}[thm]{Corollary}

\newtheorem{claim}[thm]{Claim}

\newcommand{\e}{\mathrm{e}}

\newcommand{\Exp}{\mathbb{E}}

\newcommand{\Ind}{\mathbbm{1}}

\newcommand{\dist}{\mathrm{dist}}
\newcommand{\bs}{\backslash}

\newcommand{\PHB}{\mathcal{B}}

\newcommand{\PIso}{\mathcal{IV}}

\newcommand{\PM}{\mathcal{MSW}}
\newcommand{\PSW}{\mathcal{SW}}

\newcommand{\Tmix}{T_{\mathrm{mix}}}
\newcommand{\Tcoup}{T_{\mathrm{coup}}}
\newcommand{\Trel}{T_{\mathrm{rel}}}

\newcommand{\bB}{\beta}

\newcommand{\sS}{\sigma}

\newcommand{\num}{\nu_{\mathrm{m}}}

\newcommand{\integers}{\mathbb{Z}}
\newcommand{\R}{\mathbb{R}}
\newcommand{\N}{\mathbb{N}}
\def\taumix{T_{\mathrm{mix}}}
\def\taurel{T_{\mathrm{rel}}}
\def\IC{\Omega}
\def\JC{\Omega_\textsc{J}}
\def\MJC{\Omega_\textsc{J}^\textrm{m}}
\def\MC{\mathcal{C}}
\DeclareMathOperator*{\1}{\mathbbm{1}}

\newcommand{\fptas}{\mathsf{FPTAS}}
\newcommand{\fpras}{\mathsf{FPRAS}}

\begin{document}

\title{\textbf{Swendsen-Wang Dynamics for General Graphs in the
Tree Uniqueness Region}}

\author{Antonio Blanca\thanks{School of Computer Science, Georgia
Institute of Technology, Atlanta GA 30332.\newline
 \texttt{\{ablanca3,chenzongchen,vigoda\}@gatech.edu}.
Research supported in part by NSF grants CCF-1617306 and CCF-1563838. }
\and
Zongchen Chen$^*$
\and
Eric Vigoda$^*$
}
\date{\today}

\maketitle

\begin{abstract}
	The Swendsen-Wang dynamics is a popular algorithm for sampling from the 
	Gibbs distribution for the ferromagnetic Ising model on a graph $G=(V,E)$.  The dynamics is a ``global'' Markov chain which is
	conjectured to converge to equilibrium in $O(|V|^{1/4})$ steps for any graph $G$ at any 
	(inverse) temperature $\beta$.  
	It was recently proved by Guo and Jerrum (2017) that the Swendsen-Wang dynamics 
	has polynomial mixing time on any graph at all temperatures, yet there are few 
	results providing $o(|V|)$ upper bounds on its convergence time.
	
	We prove fast convergence of the Swendsen-Wang dynamics
	on general graphs in the tree uniqueness region of the ferromagnetic Ising model.  
	In particular, when $\beta < \beta_c(d)$ where $\beta_c(d)$ denotes the uniqueness/non-uniqueness
	threshold on infinite $d$-regular trees,
	we prove that
	the \textit{relaxation time} (i.e., the inverse spectral gap) of the Swendsen-Wang dynamics
	is $\Theta(1)$	on	any graph  of maximum degree $d \ge 3$.
	Our proof utilizes a version of the Swendsen-Wang dynamics which only
	updates isolated vertices.  
	We establish that this variant of the Swendsen-Wang dynamics has mixing time $O(\log{|V|})$
	and relaxation time $\Theta(1)$ on any graph of maximum degree $d$ for all $\beta < \beta_c(d)$.
	We believe that this Markov chain may be of independent interest, as it is a \textit{monotone} Swendsen-Wang type chain. 
	As part of our proofs, we provide modest extensions of the technology of  Mossel and Sly (2013) for analyzing mixing times and of the censoring result of Peres and Winkler (2013).
	Both of these results are for the Glauber dynamics, and we extend them here 
	to general monotone Markov chains. This class of dynamics includes for example the \textit{heat-bath block dynamics}, for which we obtain new tight mixing time bounds.
\end{abstract}

\thispagestyle{empty}

\newpage

\setcounter{page}{1}

\section{Introduction}

For spin systems, sampling from the associated Gibbs distribution
is a key computational
task with a variety of applications, notably including inference/learning \cite{GRS}
and approximate counting \cite{JVV,SVV}.  
In the study of spin systems, a model of prominent interest is the Ising model.  This
is a classical model in statistical physics, which was introduced in the 1920's to study  
the ferromagnet and its physical phase transition \cite{Ising,Lenz}.
More recently, the Ising model has found numerous applications in theoretical computer science, 
computer vision, social network analysis, game theory,
biology, discrete probability and many other fields \cite{DMR,GG,Ellison,Felsenstein,MoS}.

An instance of the (ferromagnetic) Ising model is given by an undirected graph $G=(V,E)$ on $n=|V|$ vertices
and an (inverse) temperature $\beta>0$.
A configuration $\sigma \in \{+,-\}^V$
assigns a spin value ($+$~or~$-$) to each vertex $v \in V$.
The probability of a configuration $\sigma$ is proportional to 
\begin{equation}\label{eqn:gibbs}
	w(\sigma)=\exp\Big(\beta \sum_{\{v,w\}\in E} \sS(v)\sS(w) \Big),
\end{equation}
where $\sS(v)$ is the spin of $v$.
The associated Gibbs distribution $\mu=\mu_{G,\beta}$ is given by
$
\mu(\sigma) = w(\sigma)/Z,
$
where
the normalizing factor $Z$ is known as the \textit{partition function}.
Since $\beta > 0$ the system is ferromagnetic as neighboring vertices prefer to align their
spins.

For general graphs Jerrum and Sinclair \cite{JS} presented an $\fpras$
for the partition function (which yields an efficient sampler); however, its
running time is a large polynomial in $n$. Hence, there is significant interest in
obtaining tight bounds on the convergence rate of Markov chains for the Ising model,
namely, Markov chains on the space of Ising configurations $\{+,-\}^V$ that converge to
Gibbs distribution $\mu$.
A standard notion for measuring the speed of convergence to stationarity is the \textit{mixing time},
which is defined as the number of steps until the Markov chain is close to
its stationary distribution in total variation distance, starting from the worst possible initial configuration.

A simple, popular Markov chain for sampling from the Gibbs distribution is
the Glauber dynamics, commonly referred to as the Gibbs sampler
in some communities. This dynamics works by updating a 
randomly chosen vertex in each step in a reversible fashion.  
Significant progress has been 
made in understanding the mixing properties of the Glauber dynamics
and its connections to the spatial mixing (i.e., decay of correlation) properties of the underlying spin system.
In general, in the high-temperature region (small $\beta$) correlations typically
decay exponentially fast, and one expects the Glauber dynamics to converge quickly to stationarity.
For example, for the special case of the integer lattice $\integers^2$, in the high-temperature region
it is well known that the Glauber dynamics has mixing time $\Theta(n\log{n})$ \cite{MOI,Cesi,DSVW}.
For general graphs, Mossel and Sly \cite{MS} proved that the
Glauber dynamics mixes in $O(n\log{n})$ steps on any graph of maximum degree $d$ in the tree uniqueness region.
Tree uniqueness is defined as follows: 
let $T_h$ denote a (finite) complete tree of height $h$ (by complete
we mean all internal vertices have degree $d$).  
Fix the leaves to be all $+$ spins, consider the resulting conditional Gibbs
distribution on the internal vertices, and let $p^+_h$ denote the probability the root is 
assigned spin $+$ in this conditional distribution; similarly, let $p^-_h$ denote the corresponding
marginal probability with the leaves fixed to spin $-$.
When $\beta<\beta_c(d)$, where $\beta_c(d)$ is such that
\begin{equation}\label{eqn:beta_c}
(d-1)\tanh\beta_c(d) = 1,
\end{equation}
then $p^+_\infty = p^-_\infty$ 
and we say {\em tree uniqueness} 
holds since there is a unique Gibbs measure on the infinite $d$-regular tree \cite{Preston}.
In the same setting, building upon the approach of Weitz \cite{Weitz}
for the hard-core model,
Li, Lu and Yin~\cite{LLY}
provide an $\fptas$ for the partition function, but the running time is a large polynomial in $n$.

In practice, it is appealing 
to utilize non-local (or global) chains which possibly update $\Omega(n)$
vertices in a step; these chains are more popular due to their presumed speed-up
and for their ability to be naturally parallelized \cite{BL}.

A notable example for the ferromagnetic Ising model 
is the Swendsen-Wang (SW) dynamics \cite{SW}
which utilizes the random-cluster representation
to derive an elegant Markov chain in which every vertex can change its spin in every step.
The  SW dynamics works in the following manner.  
From the current spin configuration $\sS_t\in\{+,-\}^V$:
\begin{enumerate}
\item Consider the set of agreeing edges $E(\sS_t) = \{(v,w)\in E: \sS_t(v) = \sS_t(w)\}$;
\item Independently for each edge $e\in E(\sS_t)$, ``percolate'' 
by deleting $e$ with probability 
$\exp(-2\beta)$ and keeping $e$ with probability $1-\exp(-2\beta)$; this yields $F_t\subseteq E(\sS_t)$;
\item
For each connected component $C$ in the subgraph $(V,F_t)$, choose a spin $s_C$ uniformly at random
from $\{+,-\}$, and then assign spin $s_C$ to all vertices in $C$, 
yielding $\sS_{t+1}\in\{+,-\}^V$.
\end{enumerate}
\noindent The proof that the stationary distribution of the 
SW dynamics is the Gibbs distribution is non-trivial; see \cite{ES} for an elegant proof. The SW dynamics is also
well-defined for the \textit{ferromagnetic Potts model}, a natural generalization of the Ising model
that allows vertices to be assigned $q$ different spins.

The SW dynamics for the Ising model is quite appealing as it is 
conjectured to mix quickly at all temperatures.
Its behavior for the Potts model (which 
corresponds to $q>2$ spins) is more subtle, as there are
multiple examples of classes of graphs where the SW dynamics
is torpidly mixing; i.e., mixing time is exponential in the number of vertices of the graph; see, e.g., \cite{GJ,GSV,BSmf,GLP,BCT,BFKTVV}.


Despite the popularity \cite{Wang,SaSoI,SaSoII} and rich mathematical structure \cite{Grimmett}
of the SW dynamics there are few results with tight bounds on
its speed of convergence to equilibrium.  In fact, there are few results
proving the SW dynamics is faster than the Glauber dynamics 
(or the edge dynamics analog in the random-cluster
representation). Most results derive as a consequence of analyses of these local dynamics. 
Recently,  Guo and Jerrum \cite{GuoJ} established that the {mixing time} of the SW dynamics on \textit{any} graph and at any temperature 
is $O(|V|^{10})$. 
This bound, however,  
is far from the conjectured universal upper bound of $O(|V|^{1/4})$ \cite{Peres}, and once again
their result derives from a bound on a
local chain (the edge dynamics in the random-cluster representation).
  
In the special case of the \textit{mean-field} Ising model, which corresponds
to the underlying graph $G$ being the complete graph on $n$ vertices, 
Long, Nachmias, Ning and Peres \cite{LNNP} 
provided a tight analysis of the mixing time of the SW dynamics.
They prove that the mixing time of the mean-field SW dynamics is $\Theta(|V|^{1/4})$; 
this is expected to be the worst case and thus yields the aforementioned conjecture \cite{Peres}.

Another relevant case for which the speed of convergence is known
is the  two-dimensional integer lattice $\integers^2$ (more precisely, finite subsections of it).  
Blanca, Caputo, Sinclair and Vigoda \cite{BCSV} recently
established that the \textit{relaxation time} 
of the SW dynamics is $\Theta(1)$ in the high-temperature region.
The relaxation time measures the
speed of convergence to $\mu$
when the initial configuration is reasonably close to this
distribution (a so-called ``warm start'') \cite{JSV,KLS}. 
More formally, the
relaxation time is equal to the inverse spectral gap of the transition matrix of the chain and is another well-studied notion of rate of
convergence \cite{LP}.
 This result \cite{BCSV} applied
a well-established proof approach \cite{MOI,DSVW} which utilizes that $\integers^2$ is an
amenable graph.   Our goal in this paper is to establish results for general graphs of bounded degree.

Our inspiration is the result of Mossel and Sly \cite{MS} who
proved $O(n\log{n})$ mixing time of the Glauber dynamics for every graph
of maximum degree $d$.
When $\beta < \beta_c(d)$, in addition to uniqueness on the infinite $d$-regular tree, the ferromagnetic Ising model is also known to exhibit several key spatial mixing properties. 
For instance,
Mossel and Sly \cite{MS} showed that when $\beta < \beta_c(d)$ a rather strong form of spatial mixing holds on graphs of maximum degree~$d$; see Definition~\ref{dfn:prelim:sm} and Lemma~\ref{lemma:isw:SSM} in Section~\ref{sec:iso}. Using this, together with the
censoring result of Peres and Winkler \cite{PW} for the Glauber dynamics, 
they establish optimal bounds for the mixing and relaxation times of the Glauber dynamics. 
At a high-level, the censoring result \cite{PW} says that extra updates by the Markov
chain do not slow it down, and hence one can ignore transitions outside a local region
of interest in the analysis of mixing times.

A Markov chain is \textit{monotone} if it preserves the natural partial order on states; see Section~\ref{sec:background} for a detailed definition. 
We generalize the proof approach of Mossel and Sly to apply to general (non-local) monotone Markov chains.
This allows us to analyze a monotone variant of the SW dynamics, and a direct comparison of these two chains yields
a new bound  for the relaxation time of the SW dynamics.

\begin{thm}
	\label{thm:sw:intro}
	Let $G$ be an arbitrary $n$-vertex graph of maximum degree $d$.
	If $\beta < \beta_c(d)$, then the relaxation time of the Swendsen-Wang dynamics is $\Theta(1)$.
\end{thm}

\noindent
This tight  bound for the relaxation time is a substantial improvement over the best previously known $O(n)$ bound which follows from 
Ullrich's comparison theorem \cite{Ullrich1}
combined with Mossel and Sly's result \cite{MS} for the Glauber dynamics. 
We note that in Theorem \ref{thm:sw:intro}, $d$ is assumed to be a constant independent of $n$
and thus the result holds for arbitrary graphs of \textit{bounded} degree.
We also mention that
while spatial mixing properties are known to imply optimal mixing of local dynamics, only recently
the effects of these properties on the rate of convergence of non-local dynamics have started to be investigated \cite{BCSV}. 
In general, spatial mixing properties have proved to have a number of powerful algorithmic applications in the design of efficient approximation algorithms for the partition function using the associated self-avoiding walk trees (see, e.g., \cite{Weitz,SST,LLY,GK,SSSY,Sly,SlySun}).

There are three key components in our proof approach.  
First, we generalize the recursive/inductive argument of 
Mossel and Sly \cite{MS} from the Glauber dynamics to general (non-local) monotone dynamics.
Since this approach relies crucially on the censoring result of Peres and Winkler \cite{PW}
which only applies to the Glauber dynamics, 
we also need to establish   
a 
modest extension of the censoring result. For this, we use the framework of Fill and Kahn~\cite{FillK}.
Finally, we require a monotone Markov chain
that can be analyzed with these new tools
and which is naturally comparable to the SW dynamics.
To this end we utilize the \emph{Isolated-vertex dynamics} which was previously used in \cite{BCSV}.

The Isolated-vertex dynamics operates in the same manner as the SW
dynamics, except in step 3 only components of size 1 choose a new random spin (other
components keep the same spin as in $\sigma_t$). We prove that the Isolated-vertex dynamics
is \textit{monotone}.
Combining these new tools we obtain the following result.

\begin{thm}
	\label{thm:iso:intro}
	Let $G$ be an arbitrary $n$-vertex graph of maximum degree $d$.
	If $\beta < \beta_c(d)$, then the mixing time of the Isolated-vertex dynamics is $O(\log{n})$, and its relaxation time is $\Theta(1)$.
\end{thm}

\noindent
Our result for censoring may be of independent interest, as it applies to a fairly general class of non-local monotone Markov chains.
Indeed, combined with our generalization of Mossel and Sly's results \cite{MS},
it gives a general method for analyzing monotone Markov chains.

As the first application of this technology, we are able to
 establish tight bounds for the mixing and relaxation times of the \textit{block dynamics}.
  Let $\{B_1, . . . , B_r\}$ be a collection of sets (or blocks) such that $B_i \subseteq V$ and $V = \cup_i B_i$. The \textit{heat-bath block dynamics} with blocks $\{B_1, . . . , B_r\}$ is a Markov chain that in each step picks a block $B_i$ uniformly at random and updates the configuration in $B_i$ with a new configuration distributed according to the conditional measure in $B_i$ given the
 configuration in $V \setminus B_i$. 
 
 \begin{thm}
 	\label{thm:blocks:intro}
 	Let $G$ be an arbitrary $n$-vertex graph of maximum degree $d$ and let $\{B_1,\dots,B_r\}$ be an arbitrary collection of blocks such that $V = \cup_{i=1}^r B_i$.
 	If $\beta < \beta_c(d)$, then the mixing time of the block dynamics with blocks $\{B_1,\dots,B_r\}$ is $O(r \log n)$, and its relaxation time is $O(r)$.
 \end{thm}
 
 \noindent
 We observe that there are no restrictions  on the geometry of the blocks $B_i$ in the theorem other than $V = \cup_i B_i$. 
 These optimal bounds were only known before for certain specific collections of blocks.
 
As a second application of our technology, 
we consider another monotone variant of the SW dynamics, 
which we call the \textit{Monotone SW dynamics}.
This chain proceeds exactly like the SW dynamics, except that in step 3 each connected component $C$ is assigned a new random spin only with probability $1/2^{|C|-1}$ and is not updated otherwise; see Section \ref{section:msw} for a precise definition. We derive the following bounds.

\begin{thm}
	\label{thm:msw:intro}
	Let $G$ be an arbitrary $n$-vertex graph of maximum degree $d$.
	If $\beta < \beta_c(d)$, then the mixing time of the Monotone SW dynamics is $O(\log{n})$, and its relaxation time is $\Theta(1)$.
\end{thm}
  
  \noindent
 The remainder of the paper is structured as follows.
Section \ref{sec:background} contains some basic definitions and facts used throughout the paper.
 In Section \ref{sec:iso} we study the Isolated-vertex dynamics and establish Theorem~\ref{thm:iso:intro}.
 Theorem~\ref{thm:sw:intro} for the SW dynamics will
 follow as an easy corollary of these results.
 In Section \ref{sec:iso} we also state our generalization of Mossel and Sly's approach \cite{MS} for 
 non-local dynamics (Theorem \ref{thm:mixingtime:new})
 and our censoring result (Theorem \ref{thm:censoring}). 
 The proofs of these theorems are included in Sections \ref{app:main-proof} and \ref{app:censoring}, respectively.
 Finally, the full proofs of Theorems \ref{thm:blocks:intro} and \ref{thm:msw:intro} are provided in Sections \ref{app:block} and \ref{section:msw}, respectively.


\section{Background}
\label{sec:background}
In this section we provide a number of standard definitions that we will refer to in our proofs. For more details see the book \cite{LP}.

\medskip\noindent
\textbf{Ferromagnetic Ising model.}\ \ Given a graph $G=(V,E)$ and a real number $\bB>0$, the ferromagnetic Ising model on $G$ consists of the probability distribution over $\Omega_G=\{+,-\}^V$ given by
\begin{equation}
\mu_{G,\bB}\big(\sS\big) = \frac{1}{Z(G,\bB)} \exp\left[ \bB \sum\nolimits_{\{u,v\}\in E} \sigma(u)\sigma(v) \right], 
\end{equation}
where $\sS \in \Omega_G$ and
$$
Z(G,\bB) = \sum\nolimits_{\sS\in\Omega_G} \exp\left[ \bB \sum\nolimits_{\{u,v\}\in E} \sigma(u)\sigma(v) \right]
$$
is called the \textit{partition function}. 

\medskip\noindent
\textbf{Mixing and relaxation times.} \ Let $P$ be the transition matrix of an ergodic (i.e., irreducible and aperiodic) Markov chain over $\Omega_G$ with stationary distribution $\mu = \mu_{G,\beta}$.
Let $P^t(X_0,\cdot)$ denote the distribution of the chain after $t$ steps starting from $X_0 \in \Omega_G$, and let
\[
\taumix(P,\varepsilon) = \max\limits_{X_0 \in \Omega}\min \left\{ t \ge 0 : {\|{P}^t(X_0,\cdot)-\mu(\cdot)\|}_{\textsc{tv}} \le \varepsilon \right\}.
\]
The \textit{mixing time} of $P$ is defined as $\taumix(P) = \taumix(P,1/4)$. 

If $P$ is reversible with respect to (w.r.t.)\ $\mu$, the spectrum of $P$ is real. Let $1 = \lambda_1 > \lambda_2 \ge ... \ge \lambda_{|\Omega|} \geq -1$ denote its eigenvalues.  
The {\it absolute spectral gap} of $P$ is defined by $\lambda(P) = 1 - \lambda^*$, where $\lambda^* = \max\{|\lambda_2|,|\lambda_{|\Omega|}|\}$.
$\Trel(P) = \lambda(P)^{-1}$ is called 
the \textit{relaxation time} of $P$, and is another well-studied notion of rate of
convergence to $\mu$ \cite{JSV,KLS}.

\medskip\noindent
\textbf{Couplings and grand couplings.} \ A {\it (one step) coupling} of a Markov chain $\mathcal{M}$ over $\Omega_G$ specifies, for every pair of states $(X_t, Y_t) \in \Omega_G\times\Omega_G$, a probability distribution over $(X_{t+1}, Y_{t+1})$ such that the processes $\{X_t\}$ and $\{Y_t\}$, viewed in isolation, are faithful copies of $\mathcal{M}$, and if $X_t=Y_t$ then $X_{t+1}=Y_{t+1}$. 
Let $\{X_t^\sS\}_{t\geq 0}$ denote an instance of $\mathcal{M}$ started from $\sS\in\Omega_G$. 
A \emph{grand coupling} of $\mathcal{M}$ is a simultaneous coupling of  $\{X_t^\sS\}_{t\geq 0}$ for all $\sS \in \Omega_G$.

\medskip\noindent
\textbf{Monotonicity.} \
For two configurations $\sS,\tau\in\Omega_G$, we say $\sS \geq \tau$ if $\sS(v) \geq \tau(v)$ for all $v\in V$ (assuming ``$+$''$\,>\,$``$-$''). This induces a partial order on $\Omega_G$.
The ferromagnetic Ising model is \textit{monotone} w.r.t.\ this
partial order, since 
for every $B \subseteq V$ and every pair of configurations $\tau_1$, $\tau_2$ on $B$ such that $\tau_1 \ge \tau_2$ we have
$\mu(\cdot \mid \tau_1) \succeq \mu(\cdot \mid \tau_2)$, where $\succeq$
denotes stochastic domination. (For two distributions $\nu_1,\nu_2$ on $\Omega_G$, we say that $\nu_1$ stochastically dominates $\nu_2$ if for any increasing function $f \in \R^{|\Omega_G|}$ we have $\sum_{\sS\in\Omega_G} \nu_1(\sS)f(\sS) \geq \sum_{\sS\in\Omega_G} \nu_2(\sS)f(\sS)$, where a vector or function $f \in \R^{|\Omega_G|}$ is increasing if $f(\sigma) \ge f(\tau)$ for all $\sigma \ge \tau$.)

Suppose $\mathcal{M}$ is an ergodic Markov chain over $\Omega_G$ with stationary distribution $\mu$ and transition matrix $P$. A coupling of two instances $\{X_t\}$, $\{Y_t\}$ of $\mathcal{M}$ is a \textit{monotone coupling} if $X_{t+1} \ge Y_{t+1}$ whenever $X_{t} \ge Y_{t}$. We say that $\mathcal{M}$ is a \textit{monotone Markov chain} and $P$ is a \textit{monotone transition matrix} if $\mathcal{M}$ has a monotone grand coupling.

\medskip\noindent
\textbf{Comparison inequalities.} \
The \textit{Dirichlet form} of a Markov chain with transition matrix $P$
reversible w.r.t.\ $\mu$
is defined for any $f,g \in \R^{|\Omega_G|}$ as
\[
\mathcal{E}_P(f,g) = \langle f,(I-P)g \rangle_\mu = \frac{1}{2} \sum_{\sS,\tau \in \Omega_G} \mu(\sS) P(\sS,\tau) (f(\sS) - f(\tau))(g(\sS) - g(\tau)),
\]
where $\langle f,g \rangle_\mu = \sum_{\sS \in \Omega_G} \mu(\sS)f(\sS)g(\sS)$ for all $f,g \in \R^{|\Omega_G|}$.

If $P$ and $Q$ are the transition matrices of two monotone Markov chains 
reversible w.r.t.\ $\mu$, we say that $P\leq Q$ if $\langle Pf,g \rangle_\mu \leq \langle Qf,g \rangle_\mu$ for every increasing and positive $f,g \in \R^{|\Omega_G|}$. 
Note that $P\leq Q$ is equivalent to $\mathcal{E}_{P}(f,g) \geq \mathcal{E}_{Q}(f,g)$ for every increasing and positive $f,g \in \R^{|\Omega_G|}$. 

\section{Isolated-vertex dynamics}
\label{sec:iso}

In this section we consider a variant of the SW dynamics known as the \textit{Isolated-vertex dynamics} which was first introduced in \cite{BCSV}. 
We shall use this dynamics to introduce a general framework for analyzing monotone Markov chains for the Ising model and to derive our bounds for the SW dynamics. Specifically, we will prove Theorems~\ref{thm:sw:intro} and \ref{thm:iso:intro} from the introduction.

Throughout the section,
let $G=(V,E)$ be an arbitrary $n$-vertex graph of maximum degree~$d$, $\mu = \mu_{G,\beta}$ and $\Omega = \Omega_G$.
Given an Ising model configuration $\sigma_t \in \Omega$,
one step of the Isolated-vertex dynamics is given by:
\begin{enumerate}
	\item Consider the set of agreeing edges $E(\sS_t) = \{(v,w)\in E: \sigma_t(v) = \sigma_t(w)\}$;
	\item Independently for each edge $e\in E(\sS_t)$, delete $e$ with probability 
	$\exp(-2\beta)$ and keep $e$ with probability $1-\exp(-2\beta)$; this yields $F_t\subseteq E(\sS_t)$;
	\item
	For each \emph{isolated} vertex $v$ in the subgraph $(V,F_t)$ (i.e., those vertices with no incident edges in $F_t$), choose a spin uniformly at random
	from $\{+,-\}$ and assign it to $v$ to obtain $\sigma_{t+1}$; all other (non-isolated) vertices keep the same spin as in $\sS_t$.
\end{enumerate}	
We use $\PIso$ to denote the transition matrix of this chain.  The reversibility of $\PIso$ with respect to $\mu$ was established in \cite{BCSV}.
Observe also that in step 3, only \textit{isolated vertices} are updated with new random spins, whereas in the SW dynamics all connected components are assigned new random spins.
It is thus intuitive that the SW dynamics converges faster to stationarity than the Isolated-vertex dynamics. This intuition was partially captured in \cite{BCSV},  where it was proved that
\begin{equation}
\label{eq:isw:comparison}
\taurel(\PSW) \le \taurel(\PIso).
\end{equation}

The Isolated-vertex dynamics exhibits various properties
that vastly simplify its
analysis.
These properties allow us to deduce, for example, strong bounds for both its relaxation and mixing times.
Specifically, we show (in Theorem~\ref{thm:iso:intro}) that when $\beta < \beta_c(d)$, $\taumix(\PIso) = O(\log n )$  and $\taurel(\PIso) = \Theta(1)$; see (\ref{eqn:beta_c}) for the definition of $\beta_c(d)$. 
Theorem \ref{thm:sw:intro} from the introduction then follows from (\ref{eq:isw:comparison}). 


A comparison inequality like~\eqref{eq:isw:comparison} but for mixing times is not known, so Theorem \ref{thm:iso:intro} does not yield a $O(\log n)$ bound for the mixing time of the SW dynamics as one might hope. Direct comparison inequalities for mixing times are rare, since almost all known techniques 
involve
the comparison of Dirichlet forms,
and there are inherent penalties in using such inequalities to derive mixing times bounds.

The first key property of the Isolated-vertex dynamics is that, unlike the SW dynamics, this Markov chain  is \textit{monotone}.
Monotonicity is known to play a key role 
in relating
spatial mixing (i.e., decay of correlation) properties 
to fast convergence of the Glauber dynamics.
For instance, for spin systems in lattice graphs,
sophisticated functional analytic techniques are required
to establish the equivalence
between a spatial mixing property known as \textit{strong spatial mixing}
and
optimal mixing of the Glauber dynamics \cite{MOI,MOII,MOS}. 
For monotone spin systems such as the Ising model a simpler combinatorial
argument yields the same sharp result \cite{DSVW}.
This combinatorial argument is in fact more robust,
since it can be used to analyze a larger class of Markov chains, including for example the systematic scan dynamics \cite{BCSV}.

\begin{lemma}
	\label{lemma:isw:monotonicty}
	For all graphs $G$ and all $\bB>0$, the Isolated-vertex dynamics for the Ising model is monotone.
\end{lemma}


\noindent
The proof of Lemma \ref{lemma:isw:monotonicty} is given in Section \ref{subsec:IV-mono}.
The second key property of the Isolated-vertex dynamics concerns whether moves (or partial moves) of the dynamics could be \textit{censored} from the evolution of the chain without
possibly
speeding up its convergence. Censoring of Markov chains
is a well-studied notion \cite{PW,FillK,Holroyd}
that has found important applications \cite{MS,DLP,DP}.

We say that 
a stochastic $|\Omega| \times |\Omega|$ matrix $Q$ acts on a set $A \subseteq V$ if for all $\sigma,\sigma' \in \Omega$:
$$
Q(\sigma,\sigma') \neq 0 \,\,\text{iff}\,\, \sigma(V\setminus A)  = \sigma'(V\setminus A).
$$
Also recall that $P\leq P_A$ if $\langle Pf,g\rangle_\mu \leq \langle P_A f,g\rangle_\mu$ for any pair of increasing positive functions $f,g\in \mathbb{R}^{|\Omega|}$. 
\begin{defn}
	\label{dfn:isw:censoring}
	Let $G$ be an arbitrary graph and let $\beta > 0$.
	Consider an ergodic and monotone Markov chain for the Ising model on $G$, reversible w.r.t.\ $\mu=\mu_{G,\beta}$ with transition matrix $P$.
	Let $\{P_A\}_{A \subseteq V}$ be a collection of monotone stochastic matrices reversible w.r.t.\ $\mu$
	with the property that $P_A$ acts on $A$ for every $A \subseteq V$. 
	We say that  $\{P_A\}_{A \subseteq V}$  is a \emph{censoring} for $P$ if $P \le P_A$ for all $A \subseteq V$.
\end{defn}



\noindent
As an example, consider the \textit{heat-bath Glauber dynamics} for the Ising model on the graph $G=(V,E)$.
Recall that in this Markov chain a vertex $v \in V$ is chosen uniformly at random (u.a.r.)\ and a new spin is sampled for $v$ from the conditional distribution at $v$ given the configuration on $V\setminus v$.
For every $A \subseteq V$, we may take $P_A$ to be the $|\Omega| \times |\Omega|$ transition matrix of the censored heat-bath Glauber dynamics that ignores all moves outside of $A$. That is, 
if the randomly chosen vertex $v \in V$ is not in $A$, then the move is ignored; otherwise the chain proceeds as the standard heat-bath Glauber dynamics.

It is easy to check that $P_A$ is monotone and reversible w.r.t.\ $\mu$.
Moreover,
it was established in \cite{PW,FillK} that $P \le P_A$ for every $A \subseteq V$, and thus the collection $\{P_A\}_{A \subseteq V}$ is a censoring for the heat-bath Glauber dynamics. 
This particular censoring 
has been used to analyze the speed of convergence 
of  the Glauber dynamics in various settings (see \cite{MS,PW,DLP,DP}), since
it can be proved that
censored variants of the Glauber dynamics---where moves of $P$ are replaced by moves of $P_A$---converge more slowly to the stationary distribution \cite{PW,FillK}. 
Consequently, it suffices to analyze the speed of convergence of the censored chain, and this could be much simpler for suitably chosen censoring schemes.

Using the machinery from \cite{PW,FillK}, we can show that given a censoring (as defined in Definition \ref{dfn:isw:censoring}),
the strategy just mentioned for Glauber dynamics can be used for general monotone Markov chains.

\begin{thm}
	\label{thm:censoring}
	Let $G$ be an arbitrary graph and let $\beta > 0$.
	Let $\{X_t\}$ be an ergodic monotone Markov chain for the Ising model on $G$,
	reversible w.r.t.\ $\mu=\mu_{G,\beta}$ with transition matrix $P$.
	Let $\{P_A\}_{A \subseteq V}$ be a censoring for $P$
	and let $\{\hat{X}_t\}$ be a censored version of $\{X_t\}$
	that sequentially applies $P_{A_1},P_{A_2},P_{A_3}\ldots$ where $A_i \subseteq V$.
	If $X_0, Y_0$ are both sampled from a distribution $\nu$ over $\Omega$ 
	such that $\nu/\mu$ is increasing, then the following hold: 
	\begin{enumerate}
		\item $X_t \preceq \hat{X}_t$ for all $t \ge 0$;
		\item Let $\hat{P}^t = P_{A_1}\dots P_{A_t}$. Then, for all $t \ge 0$
		$$
		\|P^t(X_0,\cdot) - \mu(\cdot)\|_\textsc{tv} \le \|\hat{P}^t(X_0,\cdot) - \mu(\cdot)\|_\textsc{tv}.
		$$
	\end{enumerate}
	If $\nu/\mu$ is decreasing, then $X_t \succeq \hat{X}_t$ for all $t \ge 0$.
\end{thm}

\noindent
The proof of this theorem is provided in Section \ref{app:censoring}. 

We define next a specific censoring for  
the Isolated-vertex dynamics.
For $A \subseteq V$, let ${\PIso}_{A}$ be the transition matrix for the Markov chain that
given an Ising model configuration $\sigma_t$ generates $\sigma_{t+1}$ as follows:	
\begin{enumerate}
	\item Consider the set of agreeing edges $E(\sS_t) = \{(v,w)\in E: \sigma_t(v) = \sigma_t(w)\}$;
	\item Independently for each edge $e\in E(\sS_t)$, delete $e$ with probability 
	$\exp(-2\beta)$ and keep $e$ with probability $1-\exp(-2\beta)$; this yields $F_t\subseteq E(\sS_t)$;
	\item
	For each isolated vertex $v$ of the subgraph $(V,F_t)$ in the subset $A$, choose a spin uniformly at random
	from $\{+,-\}$ and assign it to $v$ to obtain $\sigma_{t+1}$; all other vertices keep the same spin as in $\sS_t$.
\end{enumerate}	

\begin{lemma}
	\label{lemma:isw:censoring}
	The collection of matrices $\{\PIso_{A}\}_{A \subseteq V}$ is a censoring for the Isolated-vertex dynamics.
\end{lemma}

\noindent
The proof of Lemma \ref{lemma:isw:censoring} is provided in Secion \ref{subsec:isw:censoring}.
To establish Theorem \ref{thm:iso:intro} we show that a strong form of spatial mixing, 
which is known to hold for all $\beta < \beta_c(d)$ \cite{MS}, implies
the desired mixing and relaxation times bounds for the Isolated-vertex dynamics. 
We define this notion of spatial mixing next.

For $v\in V$ and $R \in \N$, let $B(v,R)=\{u\in V: \dist(u,v)\leq R \}$ denote the ball of radius $R$ around $v$, where $\dist(\cdot,\cdot)$ denotes graph distance. Also, let $S(v,R)=B(v,R+1)\backslash B(v,R)$ be the external boundary of $B(v,R)$. 
For any $A\subseteq V$, let $\Omega_A=\{+,-\}^A$ be the set of all configurations on $A$; hence $\Omega=\Omega_G=\Omega_V$.
For $v\in V$, $u \in S(v,R)$ and $\tau \in \Omega_{S(v,R)}$, 
let $\tau_u^+$ (resp., $\tau_u^-$) be the configuration obtained from $\tau$ by changing the spin of $u$ to $+$ (resp., to $-$) and define 
\begin{equation}\label{eq:ssm:au}
a_u = \sup_{\tau\in\Omega_{S(v,R)}} \Big| \mu\left(v=+ \mid S(v,R)=\tau_u^+\right) - \mu\left(v=+ \mid S(v,R)=\tau_u^-\right) \Big|,
\end{equation}
where ``$v=+$'' represents the event that the spin of $v$ is $+$ and ``$S(v,R)=\tau_u^+$'' (resp., ``$S(v,R)=\tau_u^-$'') stands for the event that $S(v,R)$ has configuration $\tau_u^+$ (resp., $\tau_u^-$).

\begin{defn}
	\label{dfn:prelim:sm}
	We say that \emph{Aggregate Strong Spatial Mixing (ASSM)} holds for $R \in \N$, if for all $v\in V$ 
	\begin{equation*}
	\sum_{u\in S(v,R)} a_u \leq \frac{1}{4}.
	\end{equation*}
\end{defn}

\begin{lemma}[Lemma 3, \cite{MS}]\label{lemma:isw:SSM}
	For all graphs $G$ of maximum degree $d$ and all $\beta < \beta_c(d)$, there exists an integer $R=R(\beta,d) \in \N$ such that ASSM holds for $R$.	
\end{lemma}
\noindent
Theorem \ref{thm:iso:intro} is then a direct corollary of the following more general theorem.
The proof of this general theorem, which is provided in Section \ref{app:main-proof}, 
follows closely the approach in~\cite{MS} for the case of the Glauber dynamics, but
key additional considerations are required to establish such result 
for general (non-local) monotone Markov chains.
The main new innovation in our proof is the use of the more general Theorem \ref{thm:censoring}, instead of the standard censoring result in \cite{PW}.

\begin{thm}
	\label{thm:mixingtime:new}
	Let $\beta > 0$ and $G$ be an arbitrary $n$-vertex graph of maximum degree $d$ where $d$ is a constant independent of $n$.
	Consider an ergodic monotone Markov chain for the Ising model on $G$, reversible w.r.t.\ $\mu=\mu_{G,\beta}$ with transition matrix $P$. 
	Suppose $\{P_A\}_{A \subseteq V}$ is a censoring for $P$.
	If ASSM holds for a constant $R > 0$, and for any $v \in V$ and any starting configuration $\sS\in\Omega$
	\begin{equation}
	\label{eq:general-thm:main}
	\taumix(P_{B(v,R)}) \le T,
	\end{equation}
	then $\Tmix(P)=O(T \log n)$ and $\Trel(P) = O(T)$.
\end{thm}

\noindent
We note that $\taumix(P_{B(v,R)})$ denotes the mixing time from the worst possible
starting configuration, both in $B(v,R)$ and in $V \setminus B(v,R)$. (Since $P_{B(v,R)}$ only acts in $B(v,R)$, the configuration in $V \setminus B(v,R)$ remains fixed throughout the evolution of the chain and determines its stationary distribution.)

We now use Theorem \ref{thm:mixingtime:new} to establish Theorem \ref{thm:iso:intro}.
In Sections~\ref{app:block} and \ref{section:msw}, Theorem \ref{thm:mixingtime:new} is also used to establish 
Theorems \ref{thm:blocks:intro} and \ref{thm:msw:intro} from the introduction, concerning the mixing time of the block dynamics and a monotone variant of the SW dynamics.

\begin{proof}[Proof of Theorem \ref{thm:iso:intro}]
	By Lemma \ref{lemma:isw:monotonicty} the Isolated-vertex dynamics is monotone, and 
	by
	Lemma \ref{lemma:isw:censoring}
	the collection $\{\PIso_A\}_{A \subseteq V}$ is a censoring for $\PIso$.
	Moreover, 
	Lemma \ref{lemma:isw:SSM} implies that there exists a constant $R$ such that ASSM. Thus, to apply Theorem \ref{thm:mixingtime:new} all that is needed is a bound for $\taumix(\PIso_{B(v,R)})$ for all $v \in V$.
	For this, we can use a crude coupling argument. 
	Since $|B(v,R)| \le d^R$, the probability that every vertex 
	in $B(v,R)$
	becomes isolated is at least 
	$$
	{\e}^{-2\beta d |B(v,R)| } \ge {\e}^{-2\beta d^{R+1}}.
	$$
	Starting from two arbitrary configurations in $B(v,R)$, if all vertices become isolated in both configurations, then we can couple them with probability $1$. 
	Hence, we can couple two arbitrary configurations in one step with probability 
	at least $\exp(-2\beta d^{R+1})$.
	Thus, $\taumix(\PIso_{B(v,R)}) = \exp(O(\beta d^{R+1})) = O(1)$,
	and the result then follows from Theorem \ref{thm:mixingtime:new}.
\end{proof}

\begin{proof}[Proof of Theorem \ref{thm:sw:intro}]
	Follows from Theorem \ref{thm:iso:intro} and the fact that $\taurel(\PSW) \le \taurel(\PIso)$, which was established in Lemma 4.1 from \cite{BCSV}.
\end{proof}

\subsection{Monotonicity of the Isolated-vertex dynamics}
\label{subsec:IV-mono}

In this section, we show that the Isolated-vertex dynamics is monotone by constructing a monotone grand coupling; see Section \ref{sec:background} for the definition of a grand coupling.
In particular, we prove Lemma \ref{lemma:isw:monotonicty}. 

\begin{proof}[Proof of Lemma \ref{lemma:isw:monotonicty}]
	Let $\{X_t^\sS\}_{t\geq 0}$ be an instance of the Isolated-vertex dynamics starting from $\sS\in\Omega$; i.e., $X_0^\sS=\sS$. 
	We construct a grand coupling for the Isolated-vertex dynamics as follows. At time $t$:
	\begin{enumerate}
		\item For every edge $e \in E$, pick a number $r_{t}(e)$ uniformly at random from $[0,1]$;
		\item For every vertex $v\in V$, choose a uniform random spin $s_t(v)$ from $\{+,-\}$; 
		%
		\item For every $\sigma \in \Omega$: 
		\begin{enumerate}[(i)]
			\item Obtain $F_t^\sigma \subseteq E$ by including the edge $e = \{u,v\}$ in $F_t^\sigma$ iff $X_t^\sS(u) = X_t^\sS(v)$ and $r_t(e)\leq 1-\e^{-2\bB}$;
			\item For every $v\in V$, set $X_{t+1}^\sS(v) = s_t(v)$ if $v$ is an \textit{isolated vertex} in the subgraph $(V,F_t^\sigma)$; otherwise, set $X_{t+1}^\sS(v)=X_t^\sS(v)$.
		\end{enumerate}
	\end{enumerate}
	This is clearly a valid grand coupling for the Isolated-vertex dynamics. We show next that it is also monotone.
	
	
	
	Suppose $X_t^\sS \geq X_t^\tau$. We need to show that $X_{t+1}^\sS \geq X_{t+1}^\tau$ after one step of the grand coupling. 
	Let $v \in V$.
	If $v$ is not isolated in either $(V,F_t^\sS)$ or $(V,F_t^\tau)$, then the spin of $v$ remains unchanged in both $X_{t+1}^\sS$ and $X_{t+1}^\tau$, and $X_{t+1}^\sS(v)=X_t^\sS(v) \geq X_t^\tau(v)=X_{t+1}^\tau(v)$. 
	On the other hand, if $v$ is isolated in both $(V,F_t^\sS)$ and $(V,F_t^\sS)$, then the spin of $v$ is set to $s_t(v)$ in both instances of the chain; hence, $X_{t+1}^\sS(v)=s_t(v)=X_{t+1}^\tau(v)$.
	
	Suppose next that $v$ is isolated in $(V,F_t^\sS)$ but not in $(V,F_t^\tau)$.
	Then, $X_{t+1}^\sS(v)=s_t(v)$ and $X_{t+1}^\tau(v) = X_t^\tau(v)$.
	The only possibility that would violate $X_{t+1}^\sS(v) \geq X_{t+1}^\tau(v)$ is that $X_{t+1}^\sS(v) = -, X_t^\sS(v)=+$ and $X_{t+1}^\tau(v) = X_t^\tau(v) = +$. 
	If this is the case, then $X_t^\sS(v)=X_t^\tau(v)=+$.
	Moreover, since $X_t^\sS \geq X_t^\tau$, 
	all neighbors of $v$ 
	assigned ``+'' in $X_t^\tau$
	are also ``+'' in $X_t^\sigma$;
	thus
	if $v$ is isolated in $(V,F_t^\sS)$ then $v$ is also isolated in $(V,F_t^\tau)$. This leads to a contradiction,
	and so $X_{t+1}^\sS(v) \geq X_{t+1}^\tau(v)$. The case in which $v$ is isolated in $(V,F_t^\tau)$ but not in $(V,F_t^\sS)$ follows from an analogous argument. 
\end{proof}

\noindent
We can use the same grand coupling
to show that $\PIso_{A}$ is also monotone for all $A \subseteq V$.
The only required modification in the construction is that if $v \in V\setminus A$, then the spin of $v$ is not updated in either copy.
This gives the following corollary.

\begin{coro}
	\label{cor:isw:monotonicty}
	$\PIso_A$ is monotone for all $A \subseteq V$.
\end{coro}

\subsection{Censoring for the Isolated-vertex dynamics}
\label{subsec:isw:censoring}

In this section we show that the collection $\{\PIso_A\}_{A \subseteq V}$ is a censoring for
$\PIso$. Specifically, we prove Lemma \ref{lemma:isw:censoring}.

\begin{proof}[Proof of Lemma \ref{lemma:isw:censoring}]
	For all $A \subseteq V$, we need to establish that $\PIso_A$ is reversible w.r.t.\ $\mu = \mu_{G,\beta}$,
	monotone and that $\PIso \le \PIso_A$. Monotonicity follows from Corollary \ref{cor:isw:monotonicty}.
	To establish the other two facts we use an alternative representation of the matrices $\PIso$ and $\PIso_A$ that was already used in \cite{BCSV} and is inspired by the methods in \cite{Ullrich1}.
	
	Let $\JC = 2^E\times \IC$ be the \textit{joint} configuration space, where configurations consist of a spin assignment to the vertices together with a subset of the edges of $G$.
	The joint Edwards-Sokal measure $\nu$ on $\JC$ is given by
	\begin{equation}\label{eqn:joint-measure}
	\nu(F,\sigma) = \frac{1}{Z_{\textsc{j}}} p^{|F|}(1-p)^{|E \setminus F|} \1 (F \subseteq E(\sigma)),
	\end{equation}
	where 
	$p=1-{\e}^{-2\beta}$,
	$F \subseteq E$,
	$\sigma \in \IC$,
	$E(\sigma) = \{\{u,v\} \in E: \sigma(u) = \sigma(v)\}$,
	and $Z_{\textsc{j}}$ is the partition function \cite{ES}. 	
	
	Let $T$ be the $|\IC| \times |\JC|$ matrix given by:
	\begin{equation}\label{eqn:T}
		T(\sigma,(F,\tau)) = \1(\sigma=\tau)\1(F \subseteq E(\sigma)) p^{|F|} (1-p)^{|E(\sigma)\setminus F|},
	\end{equation}
	where $\sigma \in \IC$ and $(F,\tau) \in \JC$.
	The matrix $T$ corresponds to adding each edge $\{u,v\} \in E$ with $\sigma(u)=\sigma(v)$ independently with probability $p$, as in step~1 of the Isolated-vertex dynamics.
	Let $L_2(\nu)$ and $L_2(\mu)$ denote the Hilbert spaces $(\R^{|\JC|},\langle \cdot,\cdot \rangle_\nu)$ and $(\R^{|\Omega|},\langle \cdot,\cdot \rangle_\mu)$ respectively. 
	The matrix $T$ defines an operator from $L_2(\nu)$ to $L_2(\mu)$ via vector-matrix multiplication. 
	Specifically, for any $f\in\mathbb{R}^{|\JC|}$ and $\sS\in\Omega$
	\[
		Tf(\sS) = \sum_{(F,\tau)\in\JC} T(\sS,(F,\tau))f(F,\tau).
	\]
	It is easy to check that the adjoint operator $T^*:L_2(\mu) \rightarrow L_2(\nu)$ of $T$ is given by the $|\JC| \times |\IC|$ matrix 
	\begin{equation}\label{eqn:T*}
		T^*((F,\tau),\sigma) = \1(\tau = \sigma),
	\end{equation}
	with $(F,\tau) \in \JC$ and $\sigma \in \IC$.	
	Finally, for $A \subseteq V$, $F_1,F_2 \subseteq E$ and $\sigma,\tau \in \Omega$ let
	\begin{align}
	Q_{A}((F_1,\sigma),(F_2,\tau)) = \1(F_1 = F_2)\1(F_1 \subseteq E(\sigma) \cap E(\tau)) \1(\sigma(\mathcal{I}_{A}^c(F_1))=\tau(\mathcal{I}_{A}^c(F_1))) \cdot 2^{-|\mathcal{I}_{A}(F_1)|}\notag
	\end{align}
	where $\mathcal{I}_{A}(F_1)$ is the set of isolated vertices of $(V,F_1)$ in $A$ and $\mathcal{I}_{A}^c(F_1)= V\setminus\mathcal{I}_{A}(F_1)$, and similarly for $F_2$.
	For ease of notation we set $Q = Q_V$.
	It follows straightforwardly from the definition of these matrices
	that $\PIso = TQT^*$ and $\PIso_A = TQ_AT^*$ for all $A \subseteq V$. 
	It is also easy to verify that $Q = Q^2 = Q^*$, $Q_A = Q_A^2 = Q_A^*$ and that $Q = Q_AQQ_A$; see \cite{BCSV}.
	
	The reversibility of $\PIso_A$ w.r.t.\ $\mu$ follows from the fact that $\PIso_A^* = (TQ_AT^*)^* = TQ_AT^* = \PIso_A$. This implies that $\PIso_A$ is self-adjoint and thus reversible w.r.t.\ $\mu$ \cite{LP}.
	
	To establish that $\PIso \le \PIso_A$, it is sufficient to show that
	for every pair of increasing and positive functions $f_1,f_2: \R^{|\Omega|} \rightarrow \R$ on $\Omega$, we have
	\begin{equation}
	\label{eq:comparison:ineq}
	\langle f_1, \PIso f_2 \rangle_\mu \leq \langle f_1, \PIso_A f_2 \rangle_\mu.
	\end{equation}
	
	Now,
	\begin{align*}
	\langle f_1, \PIso_A f_2 \rangle_\mu 
	&= \langle f_1, TQ_AT^* f_2 \rangle_\mu
	= \langle f_1, TQ_A^2T^* f_2 \rangle_\mu
	= \langle Q_AT^* f_1, Q_AT^* f_2 \rangle_\nu = \langle \hat{f}_1,\hat{f}_2 \rangle_\nu,
	\end{align*}
	where $\hat{f}_1=Q_AT^*f_1$ and $\hat{f}_2=Q_AT^*f_2$. Similarly,
	\begin{align*}
	\langle f_1, \PIso f_2 \rangle_\mu 
	&= \langle f_1, TQ_AQ^2Q_AT^* f_2 \rangle_\mu
	= \langle QQ_AT^* f_1, QQ_AT^* f_2  \rangle_\nu = \langle Q\hat{f}_1,Q\hat{f}_2 \rangle_\nu.
	\end{align*}
	Thus, it is sufficient for us to show that $\langle Q\hat{f}_1,Q\hat{f}_2 \rangle_\nu \leq \langle \hat{f}_1,\hat{f}_2 \rangle_\nu$.

	Consider the partial order on $\JC$ where $(F,\sS) \ge (F',\sS')$ 
	iff $F=F'$ and $\sS \geq \sS'$. 
	\begin{claim}
		\label{claim:censoring:increasing}
		Suppose $f: \R^{|\IC|} \rightarrow \R$ is an increasing positive function. Then,
		$\hat{f}:\R^{|\JC|} \rightarrow \R$ where
		$\hat{f} = Q_AT^*f$ is also increasing and positive.
	\end{claim}
	
	\noindent
	Given $ \omega\in\JC$, let $\rho_\omega(\cdot) = Q(\omega,\cdot)$;
	i.e., $\rho_\omega$ is the distribution over $\JC$
	after applying $Q$ from $\omega$. We have
	\[
	Q\hat{f}_1(\omega) = \sum_{\omega' \in\JC} Q\big(\omega,\omega'\big)\hat{f}_1(\omega') = \Exp_{\rho_{\omega}}[\hat{f}_1].
	\]
	Similarly, we get $Q\hat{f}_2(\omega) = \Exp_{\rho_{\omega}}[\hat{f}_2]$. 
	
	For a distribution $\pi$ on a partially ordered set $S$, we say $\pi$ is positively correlated if for any increasing functions $f,g\in \mathbb{R}^{|S|}$ we have $\Exp_\pi[fg] \geq \Exp_\pi[f]\,\Exp_\pi[g]$. Since $\rho_{\omega}$ is a product distribution over the isolated vertices in $\omega$, $\rho_{\omega}$ is positively correlated for any $\omega \in \JC$ by Harris inequality (see, e.g., Lemma~22.14 in \cite{LP}). By Claim \ref{claim:censoring:increasing}, $\hat{f}_1$ and $\hat{f}_2$ are increasing. We then deduce that for any $\omega\in\JC$:
	\begin{align*}
	Q\hat{f}_1(\omega) \, Q\hat{f}_2(\omega) = \Exp_{\rho_{\omega}}[\hat{f}_1] \, \Exp_{\rho_{\omega}}[\hat{f}_2]
	\leq \Exp_{\rho_{\omega}}[\hat{f}_1\,\hat{f}_2].
	\end{align*}
	Putting all these facts together, we get
	\begin{align*}
	\langle Q\hat{f}_1,Q\hat{f}_2 \rangle_\nu &= \sum_{\omega \in \JC} Q\hat{f}_1(\omega) \, Q\hat{f}_2(\omega) \nu(\omega) 
	\leq \sum_{\omega \in \JC} \Exp_{\rho_{\omega}}[\hat{f}_1\,\hat{f}_2]\nu(\omega)\\
	&= \sum_{\omega,\omega' \in \JC} \hat{f}_1(\omega')\,\hat{f}_2(\omega') \rho_{\omega}(\omega')\nu(\omega)
	= \sum_{\omega,\omega' \in \JC} \hat{f}_1(\omega')\,\hat{f}_2(\omega') \rho_{\omega'}(\omega)\nu(\omega')\\
	&= \langle \hat{f}_1,\hat{f}_2 \rangle_\nu,
	\end{align*}
	where the second to last equality follows from the reversibility of $Q$ w.r.t. $\nu$; namely,
	\[
		\rho_{\omega}(\omega')\nu(\omega) = Q(\omega,\omega')\nu(\omega) = Q(\omega',\omega)\nu(\omega') = \rho_{\omega'}(\omega)\nu(\omega').
	\]
	This implies that (\ref{eq:comparison:ineq}) holds for every pair of increasing positive functions, and the theorem follows.	
\end{proof}

\noindent
We conclude this section with the proof of Claim \ref{claim:censoring:increasing}.

\begin{proof}[Proof of Claim \ref{claim:censoring:increasing}]
	 From the definition of $T^*$ we get $T^*f(F,\sS) = f(\sS)$ for any $(F,\sS)\in\JC$. 
	 Let $(F,\sS),(F,\tau)\in \JC$ be such that $\sS\geq \tau$. Then,
	\begin{align*}
	\hat{f}(F,\sS) = Q_AT^*f(F,\sS) = \sum_{(F',\sS')\in\JC} Q_A\big((F,\sS),(F',\sS')\big) f(\sS').
	\end{align*}
	Recall that $Q_A\big((F,\sS),(F',\sS')\big)>0$ iff $F=F'$ and $\sS,\sS'$ differ only in $\mathcal{I}_{A}(F)$, the set of isolated vertices in $A$. If this is the case, then
	$$ Q_A\big((F,\sS),(F,\sS')\big) = \frac{1}{2^{|\mathcal{I}_{A}(F)|}}. $$
	For $\xi\in\Omega_{\mathcal{I}_{A}(F)}$, let $\sS_\xi$ denote the configuration obtained from $\sS$ by changing the spins of vertices in $\mathcal{I}_{A}(F)$ to $\xi$; $\tau_\xi$ is defined similarly. (Recall that $\Omega_{\mathcal{I}_{A}(F)}$ denotes the set of Ising configurations on the set $\mathcal{I}_{A}(F)$.) Then, $\sigma_\xi \ge \tau_\xi$ for any $\xi\in\Omega_{\mathcal{I}_{A}(F)}$ and
	$$ \hat{f}(F,\sS) = \frac{1}{2^{|\mathcal{I}_{A}(F)|}} \sum_{\xi\in\Omega_{\mathcal{I}_{A}(F)}} f(\sS_\xi) \geq \frac{1}{2^{|\mathcal{I}_{A}(F)|}} \sum_{\xi\in\Omega_{\mathcal{I}_{A}(F)}} f(\tau_\xi) = \hat{f}(F,\tau). $$ 
	This shows that $\hat{f}$ is increasing. 
\end{proof}

\section{Proof of Theorem~\ref{thm:mixingtime:new}}
\label{app:main-proof}

In \cite{MS}, Mossel and Sly show that ASSM (see Definition \ref{dfn:prelim:sm}) implies optimal $O(n \log n)$ mixing of the Glauber dynamics on any $n$-vertex graph of bounded degree \cite{HS}.
Our proof of Theorem \ref{thm:mixingtime:new} follows the approach in \cite{MS}. The key new novelty is the use of Theorem~\ref{thm:censoring}.

\begin{proof}[Proof of Theorem \ref{thm:mixingtime:new}]
	
	
	Let $\{X_t^+\}$, $\{X_t^-\}$ be two instances of the chain such that $X_0^+ $ is the ``all plus'' configuration
	and $X_0^-$ is the ``all minus'' one. Since the chain is monotone there exists a monotone grand coupling
	of  $\{X_t^+\}$ and $\{X_t^-\}$ such that $X_t^+ \ge X_t^-$ for all $t \ge 0$. The existence of a monotone grand coupling implies that the extremal ``all plus'' and ``all minus'' are the worst possible starting configurations, and thus,
	$$ \taumix(P,\varepsilon) \le \Tcoup(\varepsilon)$$
	where $\Tcoup(\varepsilon)$ is the minimum $t$ such that $\Pr[X_t^+ \neq X_t^-] \le \varepsilon$, assuming $\{X_t^+\}$ and $\{X_t^-\}$ are coupled using the monotone coupling.
	Hence, it is sufficient to find $t$ such that for all $v \in V$
	$$
	\Pr[X_t^+(v) \neq X_t^-(v)] \le \frac{\varepsilon}{n},
	$$
	since the result would follow from a union bound over the vertices.

	Choose $R\in\mathbb{N}$ such that 
	the ASSM property holds; see Lemma~\ref{lemma:isw:SSM}.
	Let $s\in\mathbb{N}$ be arbitrary and fixed.
	For each $v\in V$, we define two instances $\{Y_t^+\}$
	and
	$\{Y_t^-\}$  of the censored chain that 
	until time $s$ evolves as the chain $P$ and after time $s$ 
	it evolves according to $P_{B(v,R)}$. 
	By assumption $P_{B(v,R)}$ is also monotone, so the evolutions of $\{Y_t^+\}$ and $\{Y_t^-\}$ can be coupled as follows:
	up to time $s$,  $\{Y_t^+\}$ and $\{Y_t^-\}$ are coupled by setting $Y_t^+ = X_t^+$ and $Y_t^- = X_t^-$ for all $0 \le t \le s$;  for $t > s$ the monotone coupling for $P_{B(v,R)}$  is used.
	Then, we have $X_t^+ \geq X_t^-$ and $Y_t^+\geq Y_t^-$ for all $t \ge 0$. 
	
	Since $P\leq P_{B(v,R)}$ by assumption, and the distribution $\nu^+$ (resp., $\nu^-$) of $X_0^+$ (resp., $X_0^-$) is such that $\nu^+/\mu$ (resp., $\nu^-/\mu$) is trivially increasing (resp., decreasing), Theorem \ref{thm:censoring} implies 
	$ Y_t^+ \succeq X_t^+$ and  $ X_t^- \succeq Y_t^-$ for all $t \ge 0$. Hence,
	$$ 
	Y_t^+ \succeq X_t^+ \succeq X_t^- \succeq Y_t^-. 
	$$
	Thus,
	\begin{align*}
	\Pr[X_t^+(v) \neq X_t^-(v)] 
	&= \Pr[X_t^+(v)=+] - \Pr[X_t^-(v)=+] \\
	&\le \Pr[Y_t^+(v)=+] - \Pr[Y_t^-(v)=+] \\
	&= \Pr[Y_t^+(v) \neq Y_t^-(v)],
	\end{align*}
	where the first and third equations follow from the monotonicity of $\{X_t^+\}$, $\{X_t^-\}$, $\{Y_t^+\}$ and $\{Y_t^-\}$ and the inequality from the fact that $Y_t^+ \succeq X_t^+$ and $Y_t^- \preceq X_t^-$.
	
		Recall our earlier definitions of $B(v,R)$ as the ball of radius $R$
	and $S(v,R)$ as the external boundary of $B(v,R)$; i.e., $B(v,R)=\{u\in V: \dist(u,v)\leq R \}$ 
	and let $S(v,R)=B(v,R+1)\backslash B(v,R)$.
	For ease of notation let $A = B(v,R+1) = B(v,R) \cup S(v,R)$
	and for $\sigma^+,\sigma^- \in \Omega_A$ let $\mathcal{F}_s(\sigma^+,\sigma^- )$ be the event
	$\{X_s^+(A) = \sS^+,X_s^-(A) = \sS^-\}$.
	Then, for $t > s$ we have
	\begin{align}
	\label{eq:triangular}
	\Pr[Y_t^+(v) \neq Y_t^-(v)\mid \mathcal{F}_s(\sigma^+,\sigma^- )]
	&\le \Big| \Pr[Y_t^+(v)=+\mid \mathcal{F}_s(\sigma^+,\sigma^- )] - \mu(v=+\mid\tau^+) \Big| \notag\\
	&+ \Big| \mu(v=+ \mid \tau^+)-\mu(v=+ \mid \tau^-) \Big|\notag\\
	&+ \Big| \Pr[Y_t^-(v)=+\mid \mathcal{F}_s(\sigma^+,\sigma^- )] - \mu(v=+\mid\tau^-) \Big|,
	\end{align}
	where 
	$\mu = \mu_{G,\beta}$,
	$\tau^+ = \sS^+(S(v,R))$ and $\tau^- = \sS^-(S(v,R))$.

	Observe that $\mu(\cdot\mid\tau^+)$ and $\mu(\cdot\mid\tau^-)$ are the stationary measures of $\{Y_t^+\}$ and $\{Y_t^-\}$ respectively, and recall that by assumption
	$$
	\max_{\sigma \in \Omega} \,\taumix(P_{B(v,R)},\sigma) \le T.
	$$
	Hence, for $t= s+T \log_4\lceil 8|A|\rceil$, we have
	\begin{equation}
	\label{eq:recursive:first}
	\Big| \Pr[Y_t^+(v)=+\mid \mathcal{F}_s(\sigma^+,\sigma^- )] - \mu(v=+\mid\tau^+) \Big| \leq \frac{1}{8|A|},
	\end{equation}
	and similarly
	\begin{equation}
	\label{eq:recursive:second}
	\Big| \Pr[Y_t^-(v)=+\mid \mathcal{F}_s(\sigma^+,\sigma^- )] - \mu(v=+\mid\tau^-) \Big| \leq \frac{1}{8|A|}.
	\end{equation}
	
	We bound next $|\mu(v=+ \mid \tau^+)-\mu(v=+ \mid \tau^-)|$.
	For $u \in S(v,R)$,
	let $a_u$ be defined as in (\ref{eq:ssm:au}) and let $S(v,R)=\{u_1,u_2,\ldots,u_l\}$ with $l=|S(v,R)|$. Let $\tau_0,\tau_1,\ldots,\tau_l$ be a sequence of configurations on $S(v,R)$ such that $\tau_j(u_k)=\tau^+(u_k)$ for $j< k\leq l$ and $\tau_j(u_k)=\tau^-(u_k)$ for $1\leq k\leq j$. That is, $\tau_0=\tau^+$, $\tau_l=\tau^-$ and $\tau_{j}$ is obtained from $\tau_{j-1}$ by changing the spin of $u_j$ from $\tau^+(u_j)$ to $\tau^-(u_j)$. The triangle inequality then implies that
	\begin{align}
	\label{eq:ssm:bound}
	\Big| \mu(v=+\mid \tau^+)-\mu(v=+\mid\tau^-) \Big| &\leq \sum_{j=1}^{l} \Big| \mu(v=+\mid\tau_{j-1})-\mu(v=+\mid\tau_{j}) \Big|\notag\\
	&\leq \sum_{j=1}^{l} \Ind\{\tau^+(u_j)\neq \tau^-(u_j)\} \cdot a_{u_j}\notag\\
	&= \sum_{u\in S(v,R)} \Ind\{\sS^+(u)\neq \sS^-(u)\} \cdot a_{u}.
	\end{align}
	Hence, plugging (\ref{eq:recursive:first}), (\ref{eq:recursive:second}) and (\ref{eq:ssm:bound}) into (\ref{eq:triangular}), we get
	\begin{align*}
	\Pr[Y_t^+(v) \neq Y_t^-(v) \mid \mathcal{F}_s(\sigma^+,\sigma^- )]  
	&\leq \frac{1}{4|A|} + \sum_{u \in S(v,R)} \Ind\{\sS^+(u)\neq \sS^-(u)\} \cdot a_{u}.
	\end{align*}
	
	Now, if $X_s^+(A) = X_s^-(A)$, then $Y_t^+(A) = Y_t^-(A)$ for all $t\geq s$.
	Therefore,
	\begin{align*}
	\Pr[Y_t^+(v) &\neq Y_t^-(v)] = \sum_{\sS^+ \neq \sS^- \in \Omega_A} \Pr[Y_t^+(v) \neq Y_t^-(v) \mid \mathcal{F}_s(\sigma^+,\sigma^- )]  \Pr[\mathcal{F}_s(\sigma^+,\sigma^- )] \\
	&\leq \frac{\Pr[X_s^+(A) \neq X_s^-(A)]}{4|A|} + \sum_{\sS^+ \neq \sS^- \in \Omega_A}  \sum_{u \in S(v,R)} \Ind\{\sS^+(u)\neq \sS^-(u)\} \cdot a_{u} \cdot \Pr[\mathcal{F}_s(\sigma^+,\sigma^- )] \\
	&=\frac{\Pr[X_s^+(A) \neq X_s^-(A)]}{4|A|} +  \sum_{u \in S(v,R)} \Pr[X_s^+(u)\neq X_s^-(u)] \cdot a_{u}.
	\end{align*}
	By union bound, 
	\begin{align*}
	\frac{\Pr[X_s^+(A) \neq X_s^-(A)]}{4|A|} &\leq \frac{1}{4|A|} \sum_{u\in A} \Pr[X_s^+(u) \neq X_s^-(u)] \leq \frac{1}{4} \max_{u\in V} \Pr[X_s^+(u)\neq X_s^-(u)].
	\end{align*}
	\noindent
	Moreover, the ASSM property (see Lemma \ref{lemma:isw:SSM}) implies that
	\begin{align*}
	\sum_{u\in S(v,R)} \Pr[X_s^+(u)\neq X_s^-(u)] \cdot a_u &\leq \max_{u\in V} \Pr[X_s^+(u)\neq X_s^-(u)] \sum_{u\in S(v,R)} a_u\\
	&\leq \frac{1}{4} \max_{u\in V} \Pr[X_s^+(u)\neq X_s^-(u)].
	\end{align*}
	Thus, we conclude that for every $v \in V$
	\begin{equation*}
	\Pr[X_t^+(v) \neq X_t^-(v)] \leq \Pr[Y_t^+(v) \neq Y_t^-(v)] \leq \frac{1}{2} \max_{u\in V} \Pr[X_s^+(u)\neq X_s^-(u)]
	\end{equation*}
	for $t= s+T \log_4\lceil 8|A|\rceil$. Taking the maximum over $v$
	\begin{equation*}
	\max_{v\in V} \Pr[X_t^+(v)\neq X_t^-(v)] \leq \frac{1}{2} \max_{v\in V} \Pr[X_s^+(v)\neq X_s^-(v)].
	\end{equation*}
	Iteratively, we get that for $\hat{T} = T \log_4\lceil 8|A|\rceil \log_2 \lceil \frac{n}{\varepsilon}\rceil$
	\begin{equation*}
	\max_{v\in V} \Pr[X_{\hat{T}}^+(v)\neq X_{\hat{T}}^-(v)] \leq \frac{\varepsilon}{n}.
	\end{equation*}
	This implies that $\taumix(P,\varepsilon) \le T \log_4\lceil 8|A|\rceil \log_2 \lceil \frac n \varepsilon \rceil$, so taking $\varepsilon=1/4$ it follows that $\Tmix(P)=O(T\log n)$ as desired. Moreover, since for $\varepsilon > 0$
	$$
	(\taurel(P)-1) \log (2\varepsilon)^{-1} \le \taumix(P,\varepsilon),
	$$
	taking $\varepsilon = n^{-1}$ yields that $\taurel(P) = O(T)$; see Theorem 12.5 in \cite{LP}.
\end{proof}

\section{Proof of Theorem~\ref{thm:censoring}}
\label{app:censoring}

\begin{proof}[Proof of Theorem~\ref{thm:censoring}]
	By assumption, $X_t$ has distribution $\nu P^t$ while $\hat{X}_t$ has distribution $\nu\hat{P}^t$ where $\hat{P}^t=P_{A_1}\dots P_{A_t}$. Since $\{P_A\}_{A\subseteq V}$ is a censoring for $P$, we have $P\leq P_A$ for all $A\subseteq V$. We show first that this implies $P^t\leq \hat{P}^t$. 
	
	Recall that $P_{A_i}$ may be viewed as an operator from $L_2(\mu)$ to $L_2(\mu)$. The reversibility of $P_{A_i}$ w.r.t. $\mu$ implies that $P_{A_i}$ is self-adjoint; i.e., $P_{A_i}^*=P_{A_i}$. Also, since $P$ is monotone, $P^k f$ is increasing for any integer $k>0$ and any increasing function $f$; see Proposition 22.7 in \cite{LP}. Combining these facts, we have that for any pair of increasing positive functions $f,g:\mathbb{R}^{|\Omega|}\to\mathbb{R}$
	\[
	\langle f, P^t g \rangle_\mu = \langle f, P(P^{t-1} g) \rangle_\mu \leq \langle f, P_{A_1}(P^{t-1} g) \rangle_\mu = \langle P_{A_1}f, P^{t-1} g \rangle_\mu.
	\]
	Note also that $P_{A_1}$ is monotone, so $P_{A_1}f$ is increasing. Iterating this argument, we obtain
	\[
	\langle f, P^t g \rangle_\mu \leq \langle P_{A_1}f, P^{t-1} g \rangle_\mu \leq \dots \leq \langle P_{A_t}\dots P_{A_1}f,g \rangle_\mu = \langle f, \hat{P}^t g \rangle_\mu.
	\]
	This shows that $P^t\leq \hat{P}^t$. 
	
	To prove $X_t\preceq \hat{X}_t$, we need to show that for any increasing function $g$
	\begin{equation}\label{eqn:censoring-proof}
	\sum_{\sS\in\Omega} \nu P^t(\sS) g(\sS) \leq \sum_{\sS\in\Omega} \nu \hat{P}^t(\sS) g(\sS).	
	\end{equation}
	Let $h:\mathbb{R}^{|\Omega|}\to\mathbb{R}$ be the function given by $h(\tau) = \nu(\tau)/\mu(\tau)$ for $\tau\in\Omega$. Then we have
	\begin{align*}
	\sum_{\sS\in\Omega} \nu P^t(\sS) g(\sS) &= \sum_{\sS\in\Omega} \Big( \sum_{\tau\in\Omega} \nu(\tau) P^t(\tau,\sS) \Big) g(\sS) = \sum_{\sS,\tau\in\Omega} \nu(\tau) P^t(\tau,\sS) g(\sS)\\
	&= \sum_{\sS,\tau\in\Omega} \mu(\tau)P^t(\tau,\sS) g(\sS)h(\tau) = \langle h, P^tg \rangle_{\mu}.
	\end{align*}
	Similarly, 
	$$ \sum_{\sS\in\Omega} \nu \hat{P}^t(\sS) g(\sS) = \langle h, \hat{P}^tg \rangle_{\mu}. $$
	The function $h$ is increasing by assumption, and thus (\ref{eqn:censoring-proof}) follows immediately from the fact that $P^t\leq \hat{P}^t$. This establishes part 1 of the theorem. Part 2 of the theorem follows from part 1 and Lemma 2.4 in \cite{PW}. 
\end{proof}

\section{Block dynamics}
\label{app:block}

As an application of the technology introduced in Section~\ref{sec:iso}, in this section we study the mixing and relaxation times of the block dynamics. Let $G=(V,E)$ be a graph of maximum degree at most $d$. Let $D=\{B_1,\ldots,B_r\}$ be a family of $r$ subsets of $V$ such that $\cup_{i=1}^r B_i=V$. Given a configuration $\sS_t\in\Omega$ at time $t$, one step of the block dynamics is given by:
\begin{enumerate}
	\item Pick $k\in\{1,2,\ldots,r\}$ uniformly at random;
	\item Sample $\sS_{t+1}(B_k)$ from $\mu(\cdot\mid\sS_t(V\bs B_k))$ and set $\sS_{t+1}(v)=\sS_t(v)$ for all $v\notin B_k$.
\end{enumerate}
\noindent
Let $\PHB^D$ be the transition matrix of the block dynamics with respect to $D$. For ease of notion, we write $\PHB=\PHB^D$ and consider the collection $D$ of blocks to be fixed. We show that when $\bB<\bB_c(d)$ the block dynamics has mixing time $O(r\log n)$ and relaxation time $O(r)$, which proves Theorem~\ref{thm:blocks:intro} from the introduction. This is done using the general framework introduced in Section~\ref{sec:iso}, which requires showing that the block dynamics is monotone, and that it has a censoring. 

The following is a standard fact about the block dynamics. 
\begin{lemma}\label{lem:block-mono}
	For all graphs $G=(V,E)$, all $\bB>0$ and any collection of blocks $D=\{B_1,\ldots,B_r\}$ such that $\cup_{i=1}^r B_i=V$, the block dynamics for the Ising model is monotone.
\end{lemma}

\begin{proof}
	A grand coupling is constructed as follows: let $\{X_t^\sS\}_{t\geq 0}$ denote the chain that starts from $\sS\in\Omega$; i.e., $X_0^\sS=\sS$. At time $t$, for all chains $\{X_t^\sS:\sS\in\Omega\}$ we choose the same uniform random block $B$ and then fix some order $\{v_1,\ldots,v_\ell\}$ of the vertices in $B$. For each $j=1,\ldots,\ell$ and $\sS\in\Omega$, the spin of $X_{t+1}^\sS(v_j)$ is sampled from the conditional distribution given $X_{t+1}^\sS(v_1),\ldots,X_{t+1}^\sS(v_{j-1})$ and $X_{t}^\sS(V\bs B)$. To update each $v_j$, we can use the standard grand coupling for the single-site Glauber dynamics; i.e., for all $\sS\in\Omega$ choose the same uniform random number $r_t(v_j)$ from $[0,1]$, and set $X_{t+1}^\sS(v_j)$ to be ``+'' if and only if 
	$$ r_t(v_j)\leq \mu\big(v_j=+\mid X_{t+1}^\sS(v_1),\ldots,X_{t+1}^\sS(v_{j-1}),X_{t}^\sS(V\bs B)\big). $$
	It is straightforward to check that this gives a monotone grand coupling.
\end{proof}

\noindent
A censoring for the block dynamics is constructed as follows. For any $A\subseteq V$, let $\PHB_A$ denote the transition matrix of the censored chain that given a configuration $\sS_t\in\Omega$ generates $\sS_{t+1}$ as follows:
\begin{enumerate}
	\item Pick $k\in\{1,2,\ldots,r\}$ uniformly at random;
	\item Sample $\sS_{t+1}(A\cap B_k)$ from $\mu(\cdot\mid\sS_t(V\bs (A\cap B_k)))$ and set $\sS_{t+1}(v)=\sS_t(v)$ for all $v\notin A\cap B_k$.
\end{enumerate}
Note that if the chosen block $B_k$ has no intersection with $A$, then the censored chain $\PHB_A$ will not update any spin in this step.

\begin{lemma}\label{lem:block-censoring}
	The collection of matrices $\{\PHB_A\}_{A\subseteq V}$ is a censoring for the block dynamics.
\end{lemma}
\noindent
Theorem~\ref{thm:blocks:intro} then follows straightforwardly from Lemma~\ref{lem:block-mono}, Lemma~\ref{lem:block-censoring} and Theorem~\ref{thm:mixingtime:new}.

\begin{proof}[Proof of Theorem~\ref{thm:blocks:intro}]
	By Lemma~\ref{lem:block-mono} the block dynamics is monotone and by Lemma~\ref{lem:block-censoring} the collection $\{\PHB_A\}_{A\subseteq V}$ is a censoring for $\PHB$. Furthermore, Lemma~\ref{lemma:isw:SSM} implies that ASSM holds for some constant $R>0$. Hence, to apply Theorem~\ref{thm:mixingtime:new} it is sufficient to bound $\Tmix(\PHB_{B(v,R)})$ for all $v\in V$. For this, we shall use a crude coupling argument. For each vertex $w$ in $B(v,R)$, the probability that $w$ is updated in one step is at least $1/r$ since at least one of the $r$ blocks contains it. Thus, $w$ is not updated after $T=\lceil r\ln(2|B(v,R)|)\rceil$ steps with probability at most 
	$$ \Big(1-\frac{1}{r}\Big)^T \leq \e^{-T/r} \leq \frac{1}{2|B(v,R)|}. $$
	Consequently, the probability that all vertices in $B(v,R)$ are updated at least once after $T$ steps is at least $1/2$ by union bound. 
	
	Now, consider two instances of the block dynamics chain with arbitrary starting configurations in $B(v,R)$ coupled as in the proof of Lemma~\ref{lem:block-mono}.
	Suppose a vertex $w\in B(v,R)$ is updated at time $t$. Observe that $w$ is set to ``$+$'' with probability at least
	\[
	\frac{\e^{-\bB d}}{\e^{\bB d}+\e^{-\bB d}} > \frac{1}{2}\e^{-2\bB d}
	\]
	and to ``$-$'' with at least the same probability. Hence, if the uniform random number $r_t(w)$ we pick in the coupling satisfies $r_t(w)\leq \frac{1}{2}\e^{-2\bB d}$, then $w$ is set to ``$+$'' in both configurations. Similarly, if $r_t(w)\geq 1-\frac{1}{2}\e^{-2\bB d}$, then $w$ is always set to ``$-$''. This implies that the spins at $w$ couple with probability at least $\e^{-2\bB d}$ each time $w$ is updated. Moreover, the event $r_t(w)\in[0,\frac{1}{2}\e^{-2\bB d}] \cup [1-\frac{1}{2}\e^{-2\bB d},1]$ is independent for distinct $w$ and $t$. Therefore, after all vertices in $B(v,R)$ are updated, the probability that the two configurations couple is at least $\e^{-2\bB d|B(v,R)|}$.
	
	Combining this with the fact that with probability at least $1/2$ all vertices in $B(v,R)$ are updated after $T$ steps, we get that the probability that two configurations couple at time $T$ is at least 
	$$\frac{1}{2} \cdot \e^{-2\bB d|B(v,R)|}=\Omega(1).$$ 
	This shows that for any $v\in V$
	$$ \Tmix(\PHB_{B(v,R)})=O(T)=O(r), $$
	and the result then follows from Theorem~\ref{thm:mixingtime:new}.
\end{proof}
\noindent
We conclude this section with the proof of Lemma \ref{lem:block-censoring}.

\begin{proof}[Proof of Lemma~\ref{lem:block-censoring}]
	The reversibility of $\PHB_A$ follows by definition for all $A\subseteq V$. Moreover, the argument in the proof of Lemma~\ref{lem:block-mono} also shows that $\PHB_A$ is monotone for all $A\subseteq V$. Thus, it suffices to show that $\PHB\leq \PHB_A$ for all $A\subseteq V$.
	
	For any $D\subseteq V$, let $R_D$ be the transition matrix corresponding to the heat-bath update in $D$; namely, sampling vertices in $D$ from the (conditional) Ising distribution conditioned on the configuration of the vertices outside $D$. Thus,
	\[
		\PHB = \frac{1}{r} \sum_{i=1}^r R_{B_i}
	\]
	and for any $f_1,f_2\in \mathbb{R}^{|\Omega|}$ we have
	\[
		\langle f_1, \PHB f_2 \rangle_\mu = \frac{1}{r} \sum_{i=1}^r \langle f_1, R_{B_i} f_2 \rangle_\mu.
	\]
	The corresponding equality holds for $\PHB_A$ and $R_{A\cap B_i}$. Hence, it suffices to show that for any $i\in\{1,\ldots,r\}$ and any increasing positive functions $f_1,f_2 \in\mathbb{R}^{|\Omega|}$ we have
	\begin{equation}\label{eqn:block-censoring}
		\langle f_1, R_{B_i} f_2 \rangle_\mu \leq \langle f_1, R_{A\cap B_i} f_2 \rangle_\mu.
	\end{equation}
	For $D\subseteq V$, the matrix $R_D$ can be viewed as an operator from $L_2(\mu)$ to $L_2(\mu)$. Recall that $R_D^*$ denotes the adjoint operator of $R_D$. The following properties of the block dynamics follow immediately from its definition: $R_{B_i}=R_{B_i}^2=R_{B_i}^*$, $R_{A\cap B_i} = R_{A\cap B_i}^2=R_{A\cap B_i}^*$ and $R_{B_i} = R_{A\cap B_i}R_{B_i}R_{A\cap B_i}$; see \cite{BCSV}.
	Let $\hat{f}_1=R_{A\cap B_i} f_1$ and $\hat{f}_2=R_{A\cap B_i} f_2$. Then
	\begin{equation*}
		\langle f_1, R_{A\cap B_i} f_2 \rangle_\mu 
		= \langle f_1, R_{A\cap B_i}^2 f_2 \rangle_\mu
		= \langle R_{A\cap B_i} f_1, R_{A\cap B_i} f_2 \rangle_\mu
		= \langle \hat{f}_1, \hat{f}_2 \rangle_\mu,
	\end{equation*}
	and
	\begin{equation*}
		\langle f_1, R_{B_i} f_2 \rangle_\mu 
		= \langle f_1, R_{A\cap B_i} R_{B_i}^2 R_{A\cap B_i} f_2 \rangle_\mu
		= \langle R_{B_i}R_{A\cap B_i} f_1, R_{B_i}R_{A\cap B_i} f_2 \rangle_\mu
		= \langle R_{B_i}\hat{f}_1, R_{B_i}\hat{f}_2 \rangle_\mu.
	\end{equation*}
	
	Given $\sS\in\Omega$, let $\rho_{\sS}(\cdot)=R_{B_i}(\sS,\cdot)$; i.e., $\rho_\sS$ is a distribution over $\Omega_{B_i}$ that results from applying the transition $R_{B_i}$ to $\sS$.
	It follows that
	$$ R_{B_i}\hat{f}_1(\sS) = \sum_{\tau\in\Omega} R_{B_i}(\sS,\tau)\hat{f}_1(\tau) = \Exp_{\rho_{\sS}} [\hat{f}_1]. $$
	Similarly, we have $R_{B_i}\hat{f}_2(\sS) = \Exp_{\rho_{\sS}} [\hat{f}_2]$. 
	Notice that the distribution $\rho_{\sS}$ is a (conditional) Ising model distribution on $B_i$, and thus it is positively correlated for any $\sS\in\Omega$; see Theorem 22.16 in \cite{LP}. 
	Also, since $R_{A\cap B_i}$ is monotone, $\hat{f}_1$ and $\hat{f}_2$ are increasing functions; see, e.g., Proposition 22.7 in \cite{LP}. This implies that for any $\sS\in\Omega$
	\[
		R_{B_i}\hat{f}_1(\sS) \, R_{B_i}\hat{f}_2(\sS) = \Exp_{\rho_{\sS}} [\hat{f}_1] \, \Exp_{\rho_{\sS}} [\hat{f}_2] \leq \Exp_{\rho_{\sS}} [\hat{f}_1 \, \hat{f}_2].
	\]
	 We then deduce that
	\begin{align*}
		\langle R_{B_i}\hat{f}_1,R_{B_i}\hat{f}_2 \rangle_\mu 
		&= \sum_{\sS\in\Omega} R_{B_i}\hat{f}_1(\sS) \, R_{B_i}\hat{f}_2(\sS) \mu(\sS)
		\leq \sum_{\sS\in\Omega} \Exp_{\rho_{\sS}} [\hat{f}_1 \, \hat{f}_2] \mu(\sS)\\
		&= \sum_{\sS,\tau\in\Omega} \hat{f}_1(\tau) \hat{f}_2(\tau) \rho_\sS(\tau)\mu(\sS)
		= \sum_{\sS,\tau\in\Omega} \hat{f}_1(\tau) \hat{f}_2(\tau) \rho_\tau(\sS)\mu(\tau)\\
		&= \langle \hat{f}_1,\hat{f}_2 \rangle_\mu
	\end{align*}
	where in the second to last equality we use the reversibility of $R_{B_i}$ w.r.t. $\mu$; namely,  
	$$ \rho_\sS(\tau)\mu(\sS) = R_{B_i}(\sS,\tau)\mu(\sS) = R_{B_i}(\tau,\sS)\mu(\tau) = \rho_\tau(\sS)\mu(\tau) $$
	for all $\sS,\tau\in\Omega$. 
	This shows that (\ref{eqn:block-censoring}) holds for any two increasing positive functions $f_1,f_2$ and the theorem follows.
\end{proof}

\section{Monotone SW dynamics}
\label{section:msw}

As a second application of our technology for analyzing monotone Markov chains, in this section we consider a monotone variant of the SW dynamics, which we call the \textit{Monotone SW dynamics}. Let $G=(V,E)$ be an arbitrary graph of maximum degree $d=O(1)$. Given an Ising configuration $\sS_t\in\Omega$ at time $t$, one step of the Monotone SW dynamics is given by:
\begin{enumerate}
	\item Consider the set of agreeing edges $E(\sS_t) = \{(v,w)\in E: \sigma_t(v) = \sigma_t(w)\}$;
	\item Independently for each edge $e\in E(\sS_t)$, delete $e$ with probability 
	$\exp(-2\beta)$ and keep $e$ with probability $1-\exp(-2\beta)$; this yields $F_t\subseteq E(\sS_t)$;
	\item For each connected component $C\subseteq V$ in the subgraph $(V,F_t)$ do the following:
	\begin{enumerate}[(i)]
		\item With probability $1/2^{|C|-1}$, choose a spin $s_C$ uniformly at random from $\{+,-\}$ and assign spin $s_C$ to all vertices in $C$; 
		\item Otherwise, with probability $1- 1/2^{|C|-1}$, every vertex in $C$ keeps the same spin as in~$\sS_t$.
	\end{enumerate}
\end{enumerate}
We denote the transition matrix of this chain by $\PM$. Notice that in step 3, all isolated vertices (components of size $1$) are updated with new random spins, and that the probability of a connected component $C$ being updated decays exponentially with the size of $C$. For comparison, recall that in the SW dynamics all connected components are updated in each step, while in the Isolated-vertex dynamics only the isolated vertices are. Using the machinery from \cite{Ullrich1}, we can show that $\PM$ is ergodic and reversible with respect to $\mu$. 
\begin{claim}\label{claim:MSW-ergodic}
	For all graphs $G$ and all $\beta>0$, the Monotone SW dynamics for the Ising model is ergodic and reversible with respect to $\mu$.
\end{claim}
\noindent
The proof of Claim~\ref{claim:MSW-ergodic} is postponed to Section \ref{sec:MSW-censoring}. We show next that, just like the Isolated-vertex dynamics, the Monotone SW dynamics also has mixing time $O(\log n)$ and relaxation time $\Theta(1)$ when $\bB<\bB_c(d)$. That is, we establish Theorem \ref{thm:msw:intro} from the introduction.
For this, 
we use the framework from Section~\ref{sec:iso} for monotone Markov chains. The first step is then to establish that the Monotone SW dynamics is indeed monotone.
\begin{lemma}
	\label{lemma:MSW-mono}
	For all graphs $G$ and all $\bB>0$, the Monotone SW dynamics for the Ising model is monotone.
\end{lemma}
\noindent
The proof of Lemma~\ref{lemma:MSW-mono} is provided in Section~\ref{sec:MSW-mono}.
Next, we construct a censoring for the Monotone SW dynamics as follows. For $A\subseteq V$, let $\PM_A$ be the transition matrix for the Markov chain that given a configuration $\sS_t\in\Omega$ generates $\sS_{t+1}$ as follows:
\begin{enumerate}
	\item Consider the set of agreeing edges $E(\sS_t) = \{(v,w)\in E: \sigma_t(v) = \sigma_t(w)\}$;
	\item Independently for each edge $e\in E(\sS_t)$, delete $e$ with probability 
	$\exp(-2\beta)$ and keep $e$ with probability $1-\exp(-2\beta)$; this yields $F_t\subseteq E(\sS_t)$;
	\item For each connected component $C\subseteq V$ of the subgraph $(V,F_t)$ that is completely contained in $A$ (i.e., $C\subseteq A$), do the following: 
	\begin{enumerate}[(i)]
		\item With probability $1/2^{|C|-1}$, choose a spin $s_C$ uniformly at random from $\{+,-\}$ and assign spin $s_C$ to all vertices in $C$; 
		\item Otherwise, with probability $1- 1/2^{|C|-1}$, every vertex in $C$ keeps the same spin as in~$\sS_t$;
	\end{enumerate}
	All other vertices keep the same spin as in $\sS_t$.
\end{enumerate}
\noindent
We emphasize that in step 3 only the components completely contained in $A$ may be assigned new random spins.

\begin{lemma}
	\label{lemma:MSW-censoring}
	The collection of matrices $\{\PM_{A}\}_{A \subseteq V}$ is a censoring for the Monotone SW dynamics.
\end{lemma}
\noindent
The proof of Lemma~\ref{lemma:MSW-censoring} is provided in Section~\ref{sec:MSW-censoring}.

Theorem \ref{thm:msw:intro} then follows immediately from Lemma~\ref{lemma:MSW-mono}, Lemma~\ref{lemma:MSW-censoring} and Theorem~\ref{thm:mixingtime:new}.

\begin{proof}[Proof of Theorem~\ref{thm:msw:intro}]
	By Lemma \ref{lemma:MSW-mono} the Monotone SW dynamics is monotone, and by Lemma \ref{lemma:MSW-censoring} the collection $\{\PM_A\}_{A \subseteq V}$ is a censoring for $\PM$. Lemma \ref{lemma:isw:SSM} shows that there exists a constant $R$ such that ASSM holds when $\bB<\bB_c(d)$ (see Section~\ref{sec:iso}). Furthermore, a crude coupling argument, analogous to that in the proof of Theorem~\ref{thm:iso:intro}, implies that $\taumix(\PM_{B(v,R)})=O(1)$ for all $v\in V$. The result then follows from Theorem \ref{thm:mixingtime:new}.
\end{proof}

\subsection{Monotonicity of the Monotone SW dynamics}\label{sec:MSW-mono}
In this section, we prove Lemma~\ref{lemma:MSW-mono}. The proof is similar to that of Lemma~\ref{lemma:isw:monotonicty} for the Isolated-vertex dynamics.

\begin{proof}[Proof of Lemma~\ref{lemma:MSW-mono}]
	Step 3 of the Monotone SW dynamics is equivalent to the following two steps:
	\begin{enumerate}
		\item[3$'$.] For each vertex $v\in V$, choose a spin $s_v$ uniformly at random from $\{+,-\}$;
		\item[4$'$.] For each connected component $C=\{v_1,\dots,v_k\}$ in the subgraph $(V,F_t)$ where $k=|C|$: 
		\begin{enumerate}[(i)]
			\item If $s_{v_1}=s_{v_2}=\dots=s_{v_k}$, then assign spin $s_{v_i}$ to $v_i$ for all $i$; 
			\item Otherwise, each $v_i$ keeps the same spin as in $\sS_t$.
		\end{enumerate}
	\end{enumerate}
	In step 4$'$, the probability that a connected component $C$ is updated with the new spin is exactly $1/2^{|C|-1}$; thus, step 3$'$ and 4$'$ are equivalent to step 3 of the Monotone SW dynamics.
	
	Let $\{X_t^\sS\}_{t\geq 0}$ be an instance of the Monotone SW dynamics starting from $\sS\in\Omega$; i.e., $X_0^\sS=\sS$. We construct a grand coupling for the Monotone SW dynamics as follows. At time $t$:
	\begin{enumerate}
		\item For every edge $e \in E$, pick a number $r_{t}(e)$ uniformly at random from $[0,1]$;
		\item For every vertex $v\in V$, choose a uniform random spin $s_t(v)$ from $\{+,-\}$; 
		\item For every $\sigma \in \Omega$:
		\begin{enumerate}[(i)]
			\item Obtain $F_t^\sigma \subseteq E$ by including the edge $e = \{u,v\}$ in $F_t^\sigma$ iff $X_t^\sS(u) = X_t^\sS(v)$ and $r_t(e)\leq 1-\e^{-2\bB}$;
			\item For every connected component $C=\{v_1,\dots,v_k\}$ in the subgraph $(V,F_t^\sS)$ and every $i=1,\dots,k$, set $X_{t+1}^\sS(v_i) = s_t(v_i)$ if $s_t(v_1)=\dots=s_t(v_k)$; otherwise, set $X_{t+1}^\sS(v_i)=X_t^\sS(v_i)$.
		\end{enumerate}
	\end{enumerate}
	We show next that this grand coupling is monotone.
	
	Suppose $X_t^\sS\geq X_t^\tau$. We need to show that $X_{t+1}^\sS \geq X_{t+1}^\tau$. Let $v\in V$ be arbitrary. If the spin of $v$ is updated in step 3(ii) in both $X_{t+1}^\sS$ and $X_{t+1}^\tau$, then $X_{t+1}^\sS(v)=s_t(v)=X_{t+1}^\tau(v)$. Similarly, if $v$ keeps its original spin in step 3(ii) in both $X_{t+1}^\sS$ and $X_{t+1}^\tau$, then $X_{t+1}^\sS (v) =X_t^\sS(v) \geq X_t^\tau(v)=X_{t+1}^\tau(v)$. 
	
	It remains to consider the cases where $v$ is updated in exactly one of $X_{t+1}^\sS$ and $X_{t+1}^\tau$. Assume first that $v$ is updated in $X_{t+1}^\sS$ but keeps its original spin in $X_{t+1}^\tau$. Suppose for sake of contradiction that $X_{t+1}^\sS(v) < X_{t+1}^\tau(v)$; i.e., $X_{t+1}^\sS(v)=-, X_{t+1}^\tau(v)=+$. Then, $X_t^\sS(v) = X_t^\tau(v)=+$ by assumption. Let $C_t^\sS(v)$ (resp., $C_t^\tau(v)$) be the connected component in the subgraph $(V,F_t^\sS)$ (resp., $(V,F_t^\tau)$) that contains $v$. Since $X_t^\sS\geq X_t^\tau$, all vertices assigned ``+'' in $X_t^\tau$ are also ``+'' in $X_t^\sS$; thus, by the way edges are coupled we have $C_t^\sS(v) \supseteq C_t^\tau(v)$. The fact that $v$ is updated in $X_{t+1}^\sS$ implies that $s_t(u)=s_t(v)$ for all $u\in C_t^\sS(v)$, and in particular, for all $u\in C_t^\tau(v)$. Then, $v$ should also be updated in $X_{t+1}^\tau$, which contradicts our assumption. The case in which $v$ is updated in $X_{t+1}^\tau$ but not in $X_{t+1}^\sS$ follows by an analogous argument.
\end{proof}
\noindent
In similar manner, we can show that $\PM_A$ is also monotone for all $A\subseteq V$.
\begin{coro}\label{coro:MSW-mono}
	$\PM_A$ is monotone for all $A\subseteq V$.
\end{coro}
\begin{proof}
	We use the grand coupling from the proof of Lemma~\ref{lemma:MSW-mono} with a slight modification. 
	Namely, in step 3(ii) a component $C$ is updated under the additional condition that $C\subseteq A$. Suppose $\{X_t^\sS\},\{X_t^\tau\}$ are two instances of $\PM_A$ starting from $\sS$ and $\tau$, respectively, and that $X_t^\sS\geq X_t^\tau$. Let $v\in V$. If $v$ is updated in both $X_{t+1}^\sS$ and $X_{t+1}^\tau$, or it is updated in neither of the two chains, then $X_{t+1}^\sS(v) \geq X_{t+1}^\tau(v)$. Now, suppose $v$ is updated in $X_{t+1}^\sS$ but not in $X_{t+1}^\tau$, and for the sake of contradiction that $X_{t+1}^\sS(v) < X_{t+1}^\tau(v)$. Then, $X_{t+1}^\sS(v)=-, X_{t+1}^\tau(v)=+$ and $X_t^\sS(v) = X_t^\tau(v)=+$. This implies that all vertices in $C_t^\sS(v)$ (the connected component in $(V,F_t^\sS)$ containing $v$) receive the same uniform random spin and that $C_t^\sS(v)\subseteq A$. Since $C_t^\tau(v) \subseteq C_t^\sS(v)$, the same property holds for $C_t^\tau(v)$. Thus, $v$ is also updated in $X_{t+1}^\tau$, leading to a contradiction. 
	The case when $v$ is updated in $X_{t+1}^\tau$ but not in $X_{t+1}^\sS$ follows by an analogous argument.
\end{proof}

\subsection{Censoring for the Monotone SW dynamics}\label{sec:MSW-censoring}
In this section we prove Lemma~\ref{lemma:MSW-censoring}. The ideas in this proof are similar to those in the proof of Lemma~\ref{lemma:isw:censoring}. Namely, we introduce a ``joint'' configuration space denoted by $\MJC$. Configurations in $\MJC$ are triples $(F,\sS,\MC)$ where $F$ is a subset of the edges, $\sS$ a spin assignment to the vertices and $\MC$ a set of ``marked'' connected components of the subgraph $(V,F)$. We show that the transition matrix $\PM$ of the Monotone SW dynamics can then be decomposed as the product of five matrices, four of which correspond to projections or liftings between the spaces $\Omega$, $\JC$ and $\MJC$ and one that corresponds to a trivial resampling in $\MJC$.

\begin{proof}[Proof of Lemma~\ref{lemma:MSW-censoring}]
	For all $A\subseteq V$, we need to show that $\PM_A$ is reversible with respect to $\mu$, monotone and that $\PM\leq \PM_A$. Monotonicity of $\PM_A$ follows from Corollary~\ref{coro:MSW-mono}. To prove the other two facts, we establish a decomposition of the matrices $\PM_A$ and $\PM$ as a product of simpler matrices, in similar fashion to what was done for the matrices $\PIso_A$ and $\PIso$ in Section~\ref{subsec:isw:censoring}. 
	
	Recall that $\JC= 2^E\times \Omega$ is the joint configuration space. 
	For $F\subseteq E$, let $\mathcal{C}(F)$ denote the set of all connected components of the subgraph $(V,F)$. We define the \textit{marked joint configuration space} $\MJC\subseteq 2^E\times \Omega\times 2^{2^V}$ by
	\[
		\MJC = \{ (F,\sS,\MC): (F,\sS)\in\JC,\MC\subseteq \mathcal{C}(F) \}.
	\]
	Connected components in $\MC$ are said to be \textit{marked}. 
	Observe that both $\PM$ and $\PM_A$ ``lift'' a configuration from $\Omega$ to one in $\JC$ (by adding the edges in step 2), which is then lifted to a configuration in $\MJC$ (by marking the components that will be updated in step 3).
	
	Let $\num$ be the \textit{marked joint measure} on $\MJC$ where each $(F,\sS,\MC)\in\MJC$ is assigned probability
	\[
		\num(F,\sS,\MC) = \nu(F,\sS) \prod_{C\in\MC} \frac{1}{2^{|C|-1}} \prod_{C\in\mathcal{C}(F)\bs\MC} \Big( 1-\frac{1}{2^{|C|-1}} \Big).
	\]
	Recall that $\nu$ is the joint Edwards-Sokal measure defined in \eqref{eqn:joint-measure}. Drawing a sample from the marked joint measure $\num$ can be achieved in the following way: first draw a sample $(F,\sS)$ from the joint measure $\nu$, and then for each connected component $C\in\mathcal{C}(F)$ independently mark $C$ (i.e., include $C$ in $\MC$) with probability $1/2^{|C|-1}$. 
	
	Let $S$ be the $|\JC|\times|\MJC|$ matrix given by
	\[
		S\big((F_1,\sS),(F_2,\tau,\MC)\big) = \1\big((F_1,\sS)=(F_2,\tau)\big) \prod_{C\in\MC} \frac{1}{2^{|C|-1}} \prod_{C\in\mathcal{C}(F_1)\bs\MC} \Big( 1-\frac{1}{2^{|C|-1}} \Big),
	\]
	where $(F_1,\sS)\in\JC$ and $(F_2,\tau,\MC)\in\MJC$. Hence, the matrix $S$ corresponds to the process of marking the connected components of a joint configuration as described above. Let $L_2(\num)$ denote the Hilbert space $(\mathbb{R}^{|\MJC|}, \langle\cdot,\cdot\rangle_{\num})$. We can view $S$ as an operator from $L_2(\num)$ to $L_2(\nu)$. The adjoint operator of $S$ can be obtained straightforwardly.
	\begin{claim}\label{claim:S*}
		The adjoint operator $S^*:L_2(\nu)\to L_2(\num)$ of $S$ is given by the $|\MJC|\times|\JC|$ matrix
		\[
		S^*\big((F_2,\tau,\MC),(F_1,\sS)\big) = \1\big((F_2,\tau)=(F_1,\sS)\big).
		\]
	\end{claim}
	\noindent
	Note that $S^*$ corresponds to dropping all the marks from the components to obtain a configuration in $\JC$.
	
	For $A\subseteq V$ and $F\subseteq E$, let $\mathcal{C}_A(F)$ be the subset of $\mathcal{C}(F)$ that contains all the connected components completely contained in $A$; let $\mathcal{C}_A^c(F)=\mathcal{C}(F)\bs \mathcal{C}_A(F)$. For $F\subseteq E$ and $\MC\subseteq \mathcal{C}(F)$, let $\MC^c=\mathcal{C}(F)\bs \MC$ be the set of all unmarked connected components. Also recall that for $\sS\in\Omega$, $E(\sigma) = \{\{u,v\} \in E: \sigma(u) = \sigma(v)\}$. 
	Given $A \subseteq V$, we define an $|\MJC|\times|\MJC|$ matrix $K_A$ indexed by the configurations of $\MJC$, which corresponds to resampling all \textit{marked} connected components that are completely contained in $A$. That is, for $(F_1,\sS,\MC_1),(F_2,\tau,\MC_2)\in\MJC$:
	\begin{align*}
		K_A\big( (F_1,\sS,\MC_1),(F_2,\tau,\MC_2) \big) ={}& 
		\1(F_1=F_2) \1(F_1\subseteq E(\sS)\cap E(\tau)) \1(\MC_1=\MC_2)\\
		&\1\big( \sS(\mathcal{C}_A^c(F_1)) = \tau(\mathcal{C}_A^c(F_1)) \big)
		\1\big( \sS(\MC_1^c) = \tau(\MC_1^c) \big) \cdot 
		2^{-|\mathcal{C}_A(F_1)\cap \MC_1|}
	\end{align*}
	where for a collection $\mathcal{U}$ of subsets of vertices, we write $\sS(\mathcal{U}) = \tau(\mathcal{U})$ if $\sS(U)=\tau(U)$ for all $U\in\mathcal{U}$.
	
	The matrix $K_A$ defines an operator from $L_2(\num)$ to $L_2(\num)$. In the following claim, which is proved later, several key properties of the matrix $K_A$ are established.
	\begin{claim}\label{claim:K*}
		For all $A \subseteq V$, $K_A=K_A^2=K_A^*$.
	\end{claim}
	\noindent
	For ease of notation we set $K = K_V$. Recall that the matrix $T$ defined in \eqref{eqn:T} is an operator from $L_2(\nu)$ to $L_2(\mu)$ and $T^*$ defined in \eqref{eqn:T*} is its adjoint operator. The following claim is an analogue of Fact~4.5 in \cite{BCSV} for the Monotone SW dynamics.
	\begin{claim}\label{claim:MSW-decomp}
		For all $A \subseteq V$, $\PM_A = TSK_AS^*T^*$.
	\end{claim}
	\noindent
	The reversibility of $\PM_A$ with respect to $\mu$ follows immediately from Claims~\ref{claim:K*} and \ref{claim:MSW-decomp} for all $A \subseteq V$: 
	\[
		\PM_A^* = (TSK_AS^*T^*)^* = TSK_AS^*T^* = \PM_A,
	\]
	so it is self-adjoint and thus reversible.
	
	To establish that $\PM \le \PM_A$, it is sufficient to show that for every pair of increasing and positive functions $f_1,f_2: \R^{|\Omega|} \rightarrow \R$ on $\Omega$, we have
	\begin{equation}\label{eq:MSW-censoring}
	\langle f_1, \PM f_2 \rangle_\mu \leq \langle f_1, \PM_A f_2 \rangle_\mu.
	\end{equation} 
	Let $\hat{f}_1=K_AS^*T^*f_1$ and $\hat{f}_2=K_AS^*T^*f_2$. Then, using Claims~\ref{claim:K*} and \ref{claim:MSW-decomp},
	\begin{equation*}
	\langle f_1, \PM_A f_2 \rangle_\mu 
	= \langle f_1, TSK_A^2S^*T^* f_2 \rangle_\mu
	= \langle K_AS^*T^* f_1, K_AS^*T^* f_2 \rangle_{\num} 
	= \langle \hat{f}_1,\hat{f}_2 \rangle_{\num}.
	\end{equation*}
	Similarly, since $K=K_AK^2K_A$ by Claim~\ref{claim:K*}, we have
	\begin{equation*}
	\langle f_1, \PM f_2 \rangle_\mu 
	= \langle f_1, TSK_AK^2K_AS^*T^* f_2 \rangle_\mu
	= \langle KK_AS^*T^* f_1, KK_AS^*T^* f_2 \rangle_{\num} 
	= \langle K\hat{f}_1,K\hat{f}_2 \rangle_{\num}.
	\end{equation*}
	Thus, it is sufficient for us to show that $\langle K\hat{f}_1,K\hat{f}_2 \rangle_{\num} \leq \langle \hat{f}_1,\hat{f}_2 \rangle_{\num}$.
	
	Consider the partial order on $\MJC$ where $(F_1,\sS,\MC_1)\geq(F_2,\tau,\MC_2)$ if and only if $F_1=F_2$, $\MC_1=\MC_2$ and $\sS\geq \tau$. The following property of the matrix $K_AS^*T^*$ will be useful.
	\begin{claim}\label{claim:hatf-increasing}
		Suppose $f:\mathbb{R}^{|\Omega|}\to\mathbb{R}$ is an increasing positive function. Then, $\hat{f}:\mathbb{R}^{|\MJC|}\to\mathbb{R}$ where $\hat{f}=K_AS^*T^*f$ is also increasing and positive.
	\end{claim}
	\noindent
	Given $\omega\in\MJC$, let $\rho_{\omega}=K(\omega,\cdot)$ be the distribution over $\MJC$ that results after applying $K$ (i.e., assigning uniform random spins to all marked connected components) from $\omega$. We get
	\[
		K\hat{f}_1(\omega) = \sum_{\omega'\in\MJC} K(\omega,\omega') \hat{f}_1(\omega') = \Exp_{\rho_{\omega}}[\hat{f}_1]
	\]
	and similarly $K\hat{f}_2(\omega)=\Exp_{\rho_{\omega}}[\hat{f}_2]$. Given $\omega$, the distribution $\rho_{\omega}$ is a product distribution over the spin assignments of all the marked connected components in $\omega$. Therefore, $\rho_{\omega}$ is positive correlated for any $\omega\in\MJC$ by Harris inequality (see, e.g., Lemma~22.14 in \cite{LP}). Since $\hat{f}_1$ and $\hat{f}_2$ are increasing by Claim~\ref{claim:hatf-increasing}, we deduce that for any $\omega\in\MJC$
	\[
		K\hat{f}_1(\omega) K\hat{f}_2(\omega) = \Exp_{\rho_{\omega}}[\hat{f}_1] \, \Exp_{\rho_{\omega}}[\hat{f}_2] \leq \Exp_{\rho_{\omega}}[\hat{f}_1\hat{f}_2].
	\]
	Hence,
	\begin{align*}
		\langle K\hat{f}_1,K\hat{f}_2 \rangle_{\num} &= \sum_{\omega\in\MJC} K\hat{f}_1(\omega) K\hat{f}_2(\omega) \num(\omega) \leq \sum_{\omega\in\MJC} \Exp_{\rho_{\omega}}[\hat{f}_1\hat{f}_2] \num(\omega)\\
		&= \sum_{\omega,\omega'\in\MJC} \hat{f}_1(\omega')\hat{f}_2(\omega') \rho_{\omega}(\omega') \num(\omega) = \sum_{\omega,\omega'\in\MJC} \hat{f}_1(\omega')\hat{f}_2(\omega') \rho_{\omega'}(\omega) \num(\omega')\\
		&\leq \langle \hat{f}_1,\hat{f}_2 \rangle_{\num}
	\end{align*}
	where the second to last equality follows from the fact that $K$ is reversible with respect to $\num$; namely,
	\[
		\rho_{\omega}(\omega') \num(\omega) = K(\omega,\omega')\num(\omega) = K(\omega',\omega)\num(\omega') = \rho_{\omega'}(\omega) \num(\omega').
	\]
	Hence, \eqref{eq:MSW-censoring} holds for every pair of increasing positive functions and the theorem follows.
\end{proof}

\subsection{Proof of auxiliary facts}
In this section we give proofs to Claims~\ref{claim:MSW-ergodic}, \ref{claim:S*}, \ref{claim:K*}, \ref{claim:MSW-decomp} and \ref{claim:hatf-increasing}.

\begin{proof}[Proof of Claim~\ref{claim:MSW-ergodic}]
	In one step of the Monotone SW dynamics, there is a positive probability that all vertices are isolated, in which case each vertex receives a uniform random spin. Thus, for any $\sS,\tau\in\Omega$ we have $\PM(\sS,\tau)>0$. This implies that the chain is ergodic (i.e., irreducible and aperiodic). The reversibility of $\PM$ with respect to $\mu$ follows from Claims~\ref{claim:K*} and \ref{claim:MSW-decomp}: $\PM^* = (TSKS^*T^*)^* = TSKS^*T^* = \PM$, so it is self-adjoint and thus reversible.
\end{proof}

\begin{proof}[Proof of Claim~\ref{claim:S*}]
	We need to show that for any $f\in\R^{|\JC|}$ and $g\in\R^{|\MJC|}$ we have $\langle f,Sg \rangle_\nu = \langle S^*f,g \rangle_{\num}$. Since
	\begin{align*}
		\langle f,Sg \rangle_\nu &= \sum_{(F_1,\sS)\in\JC} \nu(F_1,\sS) f(F_1,\sS) Sg(F_1,\sS)\\
		&= \sum_{\substack{(F_1,\sS)\in\JC\\(F_2,\tau,\MC)\in\MJC}} \nu(F_1,\sS) S\big((F_1,\sS),(F_2,\tau,\MC)\big) f(F_1,\sS)  g(F_2,\tau,\MC)
	\end{align*}
	and
	\begin{align*}
		\langle S^*f,g \rangle_{\num} &= \sum_{(F_2,\tau,\MC)\in\MJC} \num(F_2,\tau,\MC) S^*f(F_2,\tau,\MC) g(F_2,\tau,\MC)\\
		&= \sum_{\substack{(F_1,\sS)\in\JC\\(F_2,\tau,\MC)\in\MJC}} \num(F_2,\tau,\MC) S^*\big((F_2,\tau,\MC),(F_1,\sS)\big) f(F_1,\sS)  g(F_2,\tau,\MC),
	\end{align*}
	it suffices to show that for any $(F_1,\sS)\in\JC$ and $(F_2,\tau,\MC)\in\MJC$ we have
	\[
		\nu(F_1,\sS) S\big((F_1,\sS),(F_2,\tau,\MC)\big) = \num(F_2,\tau,\MC) S^*\big((F_2,\tau,\MC),(F_1,\sS)\big).
	\]
	This follows immediately from the definition of the matrices $S$ and $S^*$:
	\begin{align*}
		\nu(F_1,\sS) S\big((F_1,\sS)&,(F_2,\tau,\MC) \big)\\
		={}& \nu(F_1,\sS) \1\big((F_1,\sS)=(F_2,\tau)\big) \prod_{C\in\MC} \frac{1}{2^{|C|-1}} \prod_{C\in\mathcal{C}(F_1)\bs\MC} \Big( 1-\frac{1}{2^{|C|-1}} \Big)\\
		={}& \num(F_2,\tau,\MC) \1\big((F_1,\sS)=(F_2,\tau)\big)\\
		={}& \num(F_2,\tau,\MC) S^*\big((F_2,\tau,\MC),(F_1,\sS)\big).
	\end{align*}
	Hence, $S^*$ is the adjoint operator of $S$.
\end{proof}

\begin{proof}[Proof of Claim~\ref{claim:K*}]
	The matrix $K_A$ is symmetric for any $A\subseteq V$, and for $(F,\sS,\MC),(F,\tau,\MC)\in\MJC$ we have
	\[
		\frac{\num(F,\sS,\MC)}{\num(F,\tau,\MC)} = \frac{\nu(F,\sS)}{\nu(F,\tau)} = 1.
	\]
	Moreover, for $(F_1,\sS,\MC_1),(F_2,\tau,\MC_2)\in\MJC$ such that $K_A\big( (F_1,\sS,\MC_1),(F_2,\tau,\MC_2) \big) \neq 0$, we have $F_1=F_2$ and $\MC_1=\MC_2$. Combining these facts, we get
	\begin{equation*}
		\num(F_1,\sS,\MC_1) K_A\big( (F_1,\sS,\MC_1),(F_2,\tau,\MC_2) \big) = \num(F_2,\sS,\MC_2) K_A\big( (F_2,\sS,\MC_2),(F_1,\tau,\MC_1) \big).
	\end{equation*}
	This shows that $K_A$ is reversible with respect to $\num$ and so $K_A^*=K_A$ for all $A\subseteq V$.
	
	Since the matrix $K_A$ assigns an independent uniform random spin to each marked connected component contained in $A$, doing this process twice is equivalent to doing it once. This gives $K_A^2=K_A=K_A^*$ for all $A\subseteq V$ as claimed.
\end{proof}

\begin{proof}[Proof of Claim~\ref{claim:MSW-decomp}]
	We will prove the special case where $A=V$. The same argument works for arbitrary $A\subseteq V$. Recall that for $\sS\in\Omega$, $E(\sigma) = \{\{u,v\} \in E: \sigma(u) = \sigma(v)\}$. For any $\sS,\tau\in\Omega$, we have
	\begin{align*}
		&\PM(\sS,\tau) = \sum_{F\subseteq E(\sS)\cap E(\tau)} \Pr[(F,\sS)\mid \sS] \Pr[\tau\mid (F,\sS)],
	\end{align*}
	where
	\begin{align*}
	 \Pr[(F,\sS)\mid \sS]  &= p^{|F|} (1-p)^{|E(\sS)\bs F|}
	 \end{align*}
	 and
	 \begin{align*}
	 \Pr[\tau\mid (F,\sS)] &= \sum_{\MC\subseteq \mathcal{C}(F)} \prod_{C\in\MC} \Big(\frac{1}{2^{|C|-1}} \cdot \frac{1}{2}\Big) \prod_{C\in \mathcal{C}(F)\bs \MC} \Big(1-\frac{1}{2^{|C|-1}}\Big) \1(\sS(C)=\tau(C)).
	 \end{align*}
	Moreover, direct calculations show that for any $\sS\in\Omega$ and any $(F,\tau,\MC)\in\MJC$ we have
	\begin{align*}
		TSK\big(\sS,(F,\tau,\MC)\big)
		={}& \1(F\subseteq E(\sS)\cap E(\tau)) \1(\MC\subseteq \mathcal{C}(F))\cdot p^{|F|} (1-p)^{|E(\sigma)\setminus F|}\\
		& \cdot \prod_{C\in\MC} \Big(\frac{1}{2^{|C|-1}} \cdot \frac{1}{2}\Big) \prod_{C\in\mathcal{C}(F)\bs\MC} \Big( 1-\frac{1}{2^{|C|-1}} \Big) \1(\sS(C)=\tau(C)),
	\end{align*}
	and for any $(F,\xi,\MC)\in\MJC$, $\tau\in\Omega$
	$$S^*T^*\big((F,\xi,\MC),\tau\big)=\1(\xi=\tau).$$
	 Therefore, we deduce that for any $\sS,\tau\in\Omega$, 
	\begin{align*}
		TSKS^*T^*(\sS,\tau) &= \sum_{(F,\xi,\MC)\in\MJC} TSK\big(\sS,(F,\xi,\MC)\big) S^*T^*\big((F,\xi,\MC),\tau\big)\\
		&= \sum_{(F,\tau,\MC)\in\MJC} TSK\big(\sS,(F,\tau,\MC)\big)\\
		&= \PM(\sS,\tau).
	\end{align*}
	This implies that $\PM=TSKS^*T^*$ as claimed.
\end{proof}

\begin{proof}[Proof of Claim~\ref{claim:hatf-increasing}]
	By the definition of the matrices $T^*$ and $S^*$, we have $S^*T^* f(F,\sS,\MC) = f(\sS)$ for any $(F,\sS,\MC)\in\MJC$. Suppose $(F,\sS,\MC),(F,\tau,\MC)\in\MJC$ and $\sS\geq \tau$. Then,
	\[
		\hat{f} (F,\sS,\MC) = K_AS^*T^*f (F,\sS,\MC) = \sum_{(F',\sS',\MC')\in\MJC} K_A\big((F,\sS,\MC),(F',\sS',\MC')\big) f(\sS').
	\]
	Recall that $\mathcal{C}_A(F)$ is the set of all connected components in $(V,F)$ that are completely contained in $A$. Let 
	\[
		U_A = U_A(F,\MC) = \bigcup_{C\in \mathcal{C}_A(F) \cap \MC} C
	\]
	be the subset of vertices in the marked components completely contained in $A$. Let $\Phi_{U_A}\subseteq \{+,-\}^{U_A}$ be the set of all spin configurations on $U_A$ such that vertices from the same component receive the same spin. For $\xi\in\Phi_{U_A}$, we use $\sS_\xi$ (resp., $\tau_\xi$) to denote the configuration obtained from $\sS$ (resp., $\tau$) by replacing the spins of $U_A$ with $\xi$. Then, $\sS_\xi\geq\tau_\xi$ for any $\xi\in\Phi_{U_A}$. By definition of the matrix $K_A$, for any $(F,\sS,\MC),(F',\sS',\MC')\in\MJC$, $K_A\big((F,\sS,\MC),(F',\sS',\MC')\big)= 2^{-|\mathcal{C}_A(F)\cap \MC|}$ if and only if $F=F'$, $\MC=\MC'$ and $\sS,\sS'$ differ only in $U_A$; otherwise, it equals $0$. Thus, we deduce that
	\[
		\hat{f} (F,\sS,\MC) = 2^{-|\mathcal{C}_A(F)\cap \MC|} \sum_{\xi\in\Phi_{U_A}} f(\sS_\xi) \geq 2^{-|\mathcal{C}_A(F)\cap \MC|} \sum_{\xi\in\Phi_{U_A}} f(\tau_\xi) = \hat{f} (F,\tau,\MC).
	\]
	Hence, $\hat{f}$ is increasing for any $A\subseteq V$.
\end{proof}


%
%



                                        
\end{document}